\documentclass[10pt]{article}
\usepackage{amsmath,amssymb,amsthm,mathrsfs}
\usepackage{graphicx}
\usepackage{booktabs,multirow}
\usepackage{enumitem}
\usepackage[authoryear,longnamesfirst]{natbib}
\usepackage{xcolor}
\IfFileExists{algorithm.sty}{%
    \usepackage{algorithm}
    \usepackage{algpseudocode}
}{%
    \usepackage{float}
    \newfloat{algorithm}{tbp}{loa}
    \floatname{algorithm}{Algorithm}
    \newcommand{\algorithmicrequire}{\textbf{Require:}}
    \newcommand{\algorithmicensure}{\textbf{Ensure:}}
    \newenvironment{algorithmic}[1][]
      {\begin{enumerate}[label=\arabic*.,leftmargin=*,nosep]}
      {\end{enumerate}}
    \newcommand{\Require}{\item[]\algorithmicrequire\ }
    \newcommand{\Ensure}{\item[]\algorithmicensure\ }
    \newcommand{\State}{\item}
    \newcommand{\For}[1]{%
      \item \textbf{for }##1\textbf{ do}%
      \begin{enumerate}[label={},leftmargin=1.5em,nosep]}
    \newcommand{\EndFor}{\end{enumerate}}
}
\usepackage[colorlinks=true,allcolors=blue]{hyperref}

\newtheorem{theorem}{Theorem}
\newtheorem{lemma}[theorem]{Lemma}
\newtheorem{proposition}[theorem]{Proposition}

\newtheorem{definition}[theorem]{Definition}
\newtheorem{assumption}[theorem]{Assumption}
\newtheorem{remark}[theorem]{Remark}

\newcommand{\ConvP}{\xrightarrow{\mathbb P}}

\title{Rate-Based Causal Discovery for Heavy-Tailed Models}
\author{
Vishal Routh\\
v.routh@uga.edu\\
University of Georgia
}
\date{\today}

\begin{document}
\maketitle
\begin{abstract}
In many complex systems, such as river networks and financial markets,
extreme events propagate from causes to their effects, and the largest
observations carry the most information about the underlying causal
structure. Existing methods for causal discovery in extremes, however, are
tailored to specific mechanisms, such as linear or max-linear models. This
work develops causal discovery for heavy-tailed homogeneous structural
causal models, a flexible class that contains linear, max-linear and
norm-type mechanisms, which may differ from node to node. In these models,
the causal structure is encoded by ancestral impulse responses, which
measure how a shock at one node is transmitted to its descendants. We
propose the extremal minimum ratio estimator of these responses, which uses
only observations in which the potential cause is extreme. The estimator is
consistent under regular variation alone, and we derive its exact rates of
convergence and limiting distributions, which depend on both the causal
relation and the type of mechanism. Building on these rates, we construct
consistent estimators of the causal order and of the full set of ancestral
relations, using a threshold that adapts to the data and requires no
knowledge of the tail index. The method is illustrated on simulated data and
on river discharges in the upper Danube basin.
\end{abstract}

\section{Introduction}
\label{sec:introduction}

Causal discovery---the process of learning the underlying causal structure
from purely observational data---remains a fundamental challenge in modern
statistics and machine learning \citep{spirtes2000causation,pearl2009causality}.
Traditionally, structural causal models (SCMs) have often relied heavily on
assumptions of finite-variance noise; see, for example,
\citet{chaudhuri2025consistent} and the references therein. However, empirical
data across various complex systems frequently exhibit heavy-tailed behavior,
where extreme events play a disproportionate and highly informative role in
the system's dynamics. Such considerations arise in fields including finance
\citep{chuang2009causality}, Earth and environmental sciences
\citep{sun2021causal,mhalla2020causal}, public health
\citep{chernozhukov2011inference,zhang2012causal}, genetics
\citep{duncan2011genome}, and neuroscience \citep{zanin2016causality}, among
others.

When the noise variables driving an SCM are heavy-tailed, the asymmetry in how
extreme values propagate through the graph presents a unique opportunity.
Recent literature has demonstrated that extreme value theory, and in
particular the framework of regular variation, under which tail probabilities
decay like a power $x^{-\alpha}$ up to a slowly varying factor
\citep{resnick1987extreme,resnick2007heavy,de2006extreme,kulik2020heavy}, can
be leveraged to resolve causal directions. Existing works have successfully
explored specific instances of this phenomenon, such as heavy-tailed linear
SCMs \citep{gnecco2021causal,pasche2023causal,zhou2024efficient,krali2026causal}
and recursive max-linear models, in which each variable is the maximum,
rather than the sum, of its weighted parents and its noise
\citep{gissibl2018max,gissibl2018tail,gissibl2021identifiability,kluppelberg2021estimating,buck2021recursive,amendola2022conditional,tran2024estimating,adams2025inference},
as well as constraint-based procedures built on extremal notions of separation
\citep{jiang2025separation}; see also \citet{engelke2020graphical} for
graphical models for extremes and \citet{chavez2024causality} for a review of
causality and extremes. However, there remains a critical need for a
generalized framework capable of accommodating diverse, nonlinear functional
relationships. We also mention the recent works of \citet{fang2025structural}
and \citet{engelke2025extremes}, which formulate SCMs at the level of the
extremal limit.

In this paper, we develop a rate-based approach to causal discovery in
heavy-tailed homogeneous structural causal models (HT-HSCMs). Let
$\mathcal D=(V,E)$ be a directed acyclic graph (DAG) with node set
$V=\{1,\ldots,d\}$, and for $i\in V$ let $pa(i)$ and $an(i)$ denote the sets
of parents and of ancestors of $i$, and $An(i):=an(i)\cup\{i\}$. An HT-HSCM is
a random vector $\boldsymbol X=(X_1,\ldots,X_d)$ generated by the structural
assignments $X_i=f_i(\boldsymbol X_{pa(i)},\varepsilon_i)$, $i\in V$. The
innovations (noise variables) $\varepsilon_1,\ldots,\varepsilon_d$ are i.i.d.,
nonnegative and regularly varying with a common tail index $\alpha>0$, and
each structural function $f_i$ is continuous, coordinatewise nondecreasing and
positively homogeneous of degree one, that is, $f_i(c\boldsymbol x)=
cf_i(\boldsymbol x)$ for every $c>0$. The class was introduced by
\citet{routh2026heavy}, whose formulation required only a coordinate-deletion
property: setting some arguments of $f_i$ to zero cannot increase its value.
We strengthen this to coordinatewise monotonicity, which, unlike coordinate
deletion, is preserved under composition, and we normalize the effect of each
innovation on its own node by requiring $f_i(0,\ldots,0,1)=1$; see
Definition~\ref{Def:HT-HSCM} and Remark~\ref{rem:deletion-monotonicity}.
The sum-linear and max-linear SCMs,
$X_i=\sum_{h\in pa(i)}c_{hi}X_h+\varepsilon_i$ and
$X_i=\bigvee_{h\in pa(i)}c_{hi}X_h\vee\varepsilon_i$ with edge weights
$c_{hi}>0$, are special cases. The class also contains nonlinear mechanisms
such as $\ell_p$-aggregation,
$X_i=\|((c_{hi}X_h)_{h\in pa(i)},\varepsilon_i)\|_p$ with $0<p<\infty$,
which is the sum-linear model for $p=1$ and approaches the max-linear model as
$p\to\infty$. Moreover, the mechanism may differ from node to node.

Homogeneity makes the extremes of an HT-HSCM transparent. Substituting the
structural assignments recursively expresses each $X_i$ as a function of the
innovations $(\varepsilon_h)_{h\in An(i)}$. The \emph{ancestral impulse
response} (AIR) $F_{hi}$ of node~$i$ to node~$h$ is the value of $X_i$ when
$\varepsilon_h=1$ and all other innovations vanish. It satisfies $F_{hi}>0$ if
and only if $h\in An(i)$, and $F_{ii}=1$. By the single-big-jump principle, an
extreme observation of $\boldsymbol X$ is typically caused by one large
innovation $\varepsilon_h$, and homogeneity implies that the model transmits
it to the nodes in the fixed direction given by the $h$th row
$\boldsymbol F_{h\cdot}=(F_{h1},\ldots,F_{hd})$ of the AIR matrix
$F=(F_{hi})_{h,i\in V}$. Since the support of $F$ is the ancestral relation,
$F$ encodes the causal structure without reference to the form of the
structural functions.

\citet{routh2026heavy} studied HT-HSCMs at the population level. Let $G_i$
denote the distribution function of $X_i$. They showed that the causal tail
coefficients $\Gamma_{ij}:=\lim_{u\uparrow1}\mathbb E[G_j(X_j)\mid G_i(X_i)>u]$
of \citet{gnecco2021causal} identify the ancestral partial order of the DAG,
that is, the relation ``$i$ is an ancestor of $j$'', and that the standardized
AIRs $F_{hi}/(\sum_{k\in An(i)}F_{ki}^{\alpha})^{1/\alpha}$ can be recovered
from these coefficients by a recursive algorithm.
Estimation from data, rates of convergence and consistent structure learning
were left open. The present paper addresses these questions.

It is worth making precise what can be learned. The HT-HSCM class contains the
recursive max-linear models. Define the weight of a directed path as the
product of the edge weights $c_{uv}$ along it. In a max-linear model, an edge
$h\to i$ whose weight $c_{hi}$ is smaller than the weight of another directed
path from $h$ to $i$ can be removed without changing the distribution of
$\boldsymbol X$ \citep{gissibl2018max}. The edge set of the DAG is therefore
not identifiable in general. We target instead a causal order, that is, a
permutation $\pi$ of $V$ with $\pi(u)<\pi(v)$ whenever $u\in an(v)$, and the
reachability matrix
$R=(\mathbf 1_{\{i\in An(j)\}})_{i,j\in V}$. By the support property of the
AIRs just described, $R$ coincides with the support of $F$. Moreover, $R$
determines the transitive reduction of the DAG, which is the unique DAG with
the fewest edges having the same ancestral relations
\citep{aho1972transitive}.

Our approach rests on a deterministic inequality. For distinct nodes $i$ and
$j$, we call $(i,j)$ an \emph{ancestral pair} if $i\in an(j)$ and a
\emph{non-ancestral pair} otherwise. For all nodes $i$ and $j$,
$X_j\geq F_{ij}X_i$ almost surely (Lemma~\ref{lem:ratio-lower-bound}), so
$F_{ij}$ is a lower bound for the ratio $X_j/X_i$; for a non-ancestral pair,
$F_{ij}=0$. The bound is sharp in the extremes. Conditionally on $X_i>x$, the
law of $X_j/X_i$ converges as $x\to\infty$ to a discrete law with atoms
$F_{hj}/F_{hi}$, $h\in An(i)$, each carrying positive mass, and the
multiplicative transitivity $F_{hj}\geq F_{hi}F_{ij}$ of the AIRs
(Lemma~\ref{lem:AIR-transitivity}) shows that the smallest atom is $F_{ij}$,
attained at $h=i$. This motivates the \emph{extremal minimum ratio estimator}
(EMRE) $\widehat F_{ij}$. Given $n$ i.i.d.\ observations of $\boldsymbol X$
and an intermediate sequence $k_n$, that is, $k_n\to\infty$ and $k_n/n\to0$,
the EMRE is the smallest ratio $X_j/X_i$ among the $k_n$ observations with the
largest values of $X_i$. For recursive max-linear models,
\citet{gissibl2021identifiability} estimate the max-linear coefficients, the
max-linear analogue of the AIRs, by minima of observed ratios of components,
exploiting the fact that in these models the ratio $X_j/X_i$ of an ancestral
pair attains its lower bound with positive probability.
\citet{buck2021recursive} add multiplicative observational noise, bounded
below by one, to these models, so that the ratio no longer attains its lower
bound. When the logarithm of the noise is regularly varying at zero, the
centered and rescaled minima of the ratios over the whole sample converge to
independent Weibull limits, whose shape parameter is proportional to the
length of the max-weighted path between the two nodes. There, the rate is
governed by the behavior of the noise near its lower endpoint. An HT-HSCM has
no such noise, and outside the max-linear case its lower bound need not be
attained: in an $\ell_p$-aggregation model with $0<p<\infty$, for instance,
$X_j/X_i>F_{ij}$ almost surely whenever $i\in an(j)$. The bound is approached
in the extremes instead. Restricting the minimum to extreme observations is
therefore what makes the minimum ratio informative beyond the max-linear case,
and the resulting rate is governed by the upper tail of $X_i$ and by the
aggregation mechanism rather than by the length of a path.

The main contributions of the paper are as follows.

First, we show that the EMRE is consistent for $F_{ij}$, for every pair of
distinct nodes in every HT-HSCM, under regular variation alone.

Second, and this is our main theoretical result, we derive the exact rate of
convergence of the EMRE together with its limiting distribution. This requires
two conditions near zero on the innovations: a lower-tail condition
$\mathbb P(\varepsilon_1\leq t)=o(t^{\alpha/(1-\delta)})$ as $t\downarrow0$,
and a negative-moment condition $\mathbb E[\varepsilon_1^{-(\alpha+\delta')}]
<\infty$, for some $\delta,\delta'\in(0,1)$. The first implies the second,
and both hold, for example, for Pareto and Fr\'echet innovations. For $t>1$,
let $U_i(t)$ denote the quantile of order $1-1/t$ of $X_i$, and put
$a_n:=U_i(n)$. Three regimes emerge. If $i\notin An(j)$, then
$a_n\widehat F_{ij}$ converges in distribution to a Weibull law with shape
parameter $\alpha$ and an explicit scale. The scale is determined by the
negative moment $\mathbb E[X_j^{-\alpha}]$ and by the fraction
$\sum_{h\in An(i)\setminus An(j)}F_{hi}^{\alpha}/\sum_{h\in An(i)}F_{hi}^{\alpha}$
of the tail of $X_i$ that is due to ancestors of $i$ that are not ancestors
of $j$. This result holds in every HT-HSCM, including models whose mechanisms
differ across nodes. If $i\in an(j)$ in an $\ell_p$-aggregation model with
$0<p<\infty$, then $a_n^p(\widehat F_{ij}-F_{ij})$ converges in distribution
to a Weibull law with shape parameter $\alpha/p$, again with an explicit
scale. In the max-linear model, $\widehat F_{ij}=F_{ij}$ with probability
tending to one. If $X_i$ has a Pareto-type tail,
$\mathbb P(X_i>x)\sim cx^{-\alpha}$ for some $c>0$, then
$a_n\sim(cn)^{1/\alpha}$, and the three regimes correspond to the rates
$n^{-1/\alpha}$ and $n^{-p/\alpha}$ and to exact recovery. The two rates have
transparent sources. A non-ancestral ratio is small essentially only when
$X_i$ is driven by an innovation that does not reach $j$, while $X_j$ remains
moderate. For an ancestral pair in an $\ell_p$-aggregation model,
$X_j^p-F_{ij}^pX_i^p$ is a nonnegative random variable that does not involve
the innovations $\varepsilon_h$ with $F_{hj}=F_{ij}F_{hi}$, which are exactly
those whose extremes produce ratios $X_j/X_i$ close to $F_{ij}$; hence the
excess $X_j/X_i-F_{ij}$ decays like $X_i^{-p}$. The rate is therefore governed
jointly by the causal relation and by the aggregation mechanism: the ancestral
error decays more slowly than the non-ancestral estimate when $p<1$ and faster
when $p>1$, and the max-linear model is the limiting case of exact recovery.
None of these limit laws depends on the choice of $k_n$.

Third, we use the EMRE for structure learning. Supplying the estimated AIRs to
the extremal ancestral search (EASE) algorithm of \citet{gnecco2021causal},
originally proposed with the causal tail coefficient for heavy-tailed linear
SCMs, yields a consistent estimator of a causal order for the entire HT-HSCM
class, again under regular variation alone. EASE builds the order greedily: at
each step it selects, among the remaining nodes, the one whose largest
estimated AIR from the other remaining nodes is smallest. Given this order, we
estimate the reachability matrix by testing only pairs that point forward in
the order and declaring $i\in an(j)$ when
$\widehat F_{ij}>X_{i,(n-k_n)}^{-1}$, where $X_{i,(n-k_n)}$, the
$(k_n+1)$st largest of the $n$ observations of $X_i$, is an intermediate
order statistic. This threshold is \emph{rate-based}. It tends to zero, so
ancestral pairs, whose EMRE converges to a positive AIR, are retained; but it
is of order $U_i(n/k_n)^{-1}$, which exceeds the order $U_i(n)^{-1}$ of a
non-ancestral EMRE by a factor that tends to infinity, roughly like
$k_n^{1/\alpha}$, so non-ancestral pairs are eventually rejected. We prove that
the resulting estimator recovers $R$ with probability tending to one. It
requires no knowledge of the tail index or of the structural functions, which
may differ across nodes, and the EMRE and the threshold may be computed with
different intermediate sequences.

Fourth, a simulation study with DAGs of up to $30$ nodes, in which every node
independently follows a sum-linear, Euclidean-norm ($\ell_2$-aggregation) or
max-linear mechanism, shows that the procedure performs well in finite
samples. At the reported tuning parameters, the power (the proportion of
ancestral pairs detected) is at least $0.96$, the normalized structural Hamming
distance (the proportion of misclassified off-diagonal entries of $R$) is at
most $0.031$, and the ancestral violation rate (the proportion of ancestral
pairs placed in the wrong order by the estimated causal order) does not exceed
$10^{-4}$. The study also explains how the two tuning parameters, namely the
intermediate sequences used in the EMRE and in the threshold, should be chosen
and why recovery is harder for $\alpha=2$ than for $\alpha=1$.

Fifth, we apply the procedure to river discharges at twelve gauging stations
of the upper Danube basin, whose flow network provides a known causal
structure, and study how the results depend on the two intermediate
sequences. After declustering the daily discharges into $430$ summer flood
events and standardizing the margins, EASE recovers a causal order without
ancestral violations whenever the EMRE uses between $8$ and $28$ events. Over
a band of threshold exponents, the reachability estimator attains a normalized
structural Hamming distance between $0.098$ and $0.114$, compared with $0.136$
for the empty graph. Its remaining errors have a physical explanation: false
ancestral relations arise between adjacent catchments that are hit by the same
storms, a common cause that lies outside the model, and missed relations
involve small impulse responses into the most downstream station.

Our results complement the existing literature in several ways. The causal
tail coefficient is defined through the marginal distribution functions, and
\citet{routh2026heavy} recover standardized AIRs from it through a
population-level recursion. The EMRE instead estimates each AIR directly, on
the scale of the data, and comes with exact rates and limit laws. In the
max-linear case, it recovers the AIR exactly with probability tending to one,
mirroring the exact identification of max-linear coefficients in
\citet{gissibl2021identifiability}, while for other mechanisms it remains
consistent, at a rate dictated by the mechanism. Our approach to structure
learning also differs from that of \citet{buck2021recursive}, who estimate a
topological order by sorting minimum ratios and set small estimated
coefficients to zero with a user-chosen threshold. Their limit theory concerns
pairs with a positive coefficient. In HT-HSCMs, we also derive the rate for
non-ancestral pairs, and it is this rate that justifies a threshold that tends
to zero. Finally, unlike the formulations of \citet{fang2025structural} and
\citet{engelke2025extremes}, which are set at the level of the extremal limit,
the structural equations of an HT-HSCM hold for the observations themselves,
and our estimators act directly on ratios of observed variables.

The rest of the paper is organized as follows. Section~\ref{sec:prelim}
collects the graph-theoretic and regular-variation preliminaries, defines
HT-HSCMs and the AIR matrix, derives the joint tail of an HT-HSCM, and presents
the sum-linear, max-linear and $\ell_p$-aggregation examples together with the
basic properties of the AIRs. Section~\ref{sec:estimator} introduces the EMRE
and establishes its consistency, rates of convergence and limiting
distributions. Section~\ref{sec:structure} applies the EMRE to the estimation
of a causal order and of the reachability matrix. Section~\ref{sec:simulation}
reports the simulation study, and Section~\ref{sec:danube-sensitivity} applies
the procedure to the upper Danube data. The supplementary material contains
additional numerical results.

\section{Preliminaries}
\label{sec:prelim}

\subsection{Graph notions}
Let $\mathcal{D}=(V,E)$ be a directed acyclic graph (DAG) with node set
$V=\{1,\ldots,d\}$ and edge set $E\subseteq V\times V$.
Node $j$ is a parent of node $i$ if $(j,i)\in E$, and we write
\[
pa(i)=\{j\in V:(j,i)\in E\}, \qquad j\to i \iff (j,i)\in E.
\]
For an SCM defined on a DAG, each node in \(V\) corresponds to a random
variable, and a directed edge indicates that the variable at the child node
functionally depends on the variable at the parent node.

By a directed path from $j$ to $i$, we mean a sequence
$(k_0,\ldots,k_s)$ with $s\geq1$, $k_0=j$, $k_s=i$, and
$k_\ell\to k_{\ell+1}$ for $\ell=0,\ldots,s-1$. If such a path exists,
we write $j\leadsto i$. We then call $j$ an ancestor of $i$ and $i$ a
descendant of $j$. Let $an(i)$ and $de(i)$ denote the strict ancestors and
strict descendants of $i$, respectively, and set
\[
An(i)=an(i)\cup\{i\}, \qquad De(i)=de(i)\cup\{i\}.
\]

While the graph-theoretic notation above provides the combinatorial structure
of our causal models, the statistical identifiability of these structures
relies on the specific tail behavior of the associated random variables. In
particular, to resolve causal directions from observational data, we must
characterize how extreme values propagate through the network. This
necessitates a formal treatment of regular variation, which provides the
mathematical language to describe the heavy-tailed noise distributions that
drive our proposed HT-HSCM. We detail these probabilistic preliminaries in the
following subsection.

\subsection{Regularly varying functions and random variables}
Throughout, for positive functions $f$ and $g$, we write $f\sim g$ if
$f(x)/g(x)\to1$ as $x\to\infty$.

A positive measurable function $f$, defined on $[x_0,\infty)$ for some
$x_0>0$, is \emph{regularly varying} (at infinity) with index
$\rho\in\mathbb R$, written $f\in RV_\rho$, if
\begin{equation}\label{eq:RV-function}
    \lim_{x\to\infty}\frac{f(cx)}{f(x)}=c^{\rho}
    \qquad\text{for every }c>0 .
\end{equation}
If $\rho=0$, then $f$ is called \emph{slowly varying}. Every
$f\in RV_\rho$ can be written as $f(x)=x^{\rho}\ell(x)$ with
$\ell\in RV_0$.

For a random variable, regular variation refers to the upper tail. A
nonnegative random variable $X$ with distribution function $G_X$ is
\emph{regularly varying with tail index $\alpha>0$}, written
$X\in RV_+(\alpha)$, if its survival function $1-G_X$ belongs to
$RV_{-\alpha}$, that is,
\[
    \mathbb P(X>x)\sim\ell(x)x^{-\alpha},\qquad x\to\infty,
\]
for some $\ell\in RV_0$. Classical examples are the Pareto and Fr\'echet
distributions. The absolute values of Student's $t$ and Cauchy random
variables are further examples, with tail index equal to the number of
degrees of freedom and to $1$, respectively.

The tail of $X$ can equivalently be described through its quantiles. Let
$G_X^{\leftarrow}(y):=\inf\{x\in\mathbb R:G_X(x)\geq y\}$ denote the quantile
function of $X$, and define the \emph{upper quantile function}
\begin{equation}\label{eq:upper-quantile}
    U_X(t):=G_X^{\leftarrow}\Bigl(1-\frac1t\Bigr),\qquad t>1,
\end{equation}
so that $U_X(t)$ is the quantile of order $1-1/t$. Then
$X\in RV_+(\alpha)$ if and only if $U_X\in RV_{1/\alpha}$, and in that case
the sequence $c_n:=U_X(n)$ satisfies
\begin{equation}\label{eq:RV-sequence}
    \lim_{n\to\infty}n\,\mathbb P(X>c_nt)=t^{-\alpha}
    \qquad\text{for every }t>0
\end{equation}
\citep[Section~1.3]{kulik2020heavy}. The normalizing sequence matters only
up to a constant factor. If $c_n'\sim Cc_n$ for some $C>0$, then
\eqref{eq:RV-sequence} and the monotonicity of $t\mapsto\mathbb P(X>t)$ give
\[
    \lim_{n\to\infty}n\,\mathbb P(X>c_n't)=(Ct)^{-\alpha}
    \qquad\text{for every }t>0 .
\]
The quantile functions $U_i:=U_{X_i}$ of the nodes and the sequence
$a_n=U_1(n)$ used in Sections~\ref{sec:estimator} and~\ref{sec:structure}
are of this form.

Ratios such as $U_i(n)/U_i(n/k_n)$ or $\mathbb P(X>tx)/\mathbb P(X>t)$ must
often be controlled uniformly in $x$, not only in the limit
\eqref{eq:RV-function}. Potter's bounds provide this control.

\begin{lemma}[Potter's bounds]\label{Potter's bounds}
Let $f\in RV_\rho$ for some $\rho\in\mathbb R$. For every $\epsilon>0$
there exists $t_0>0$ such that
\[
    (1-\epsilon)\,x^{\rho-\epsilon}
    \leq\frac{f(tx)}{f(t)}
    \leq(1+\epsilon)\,x^{\rho+\epsilon}
    \qquad\text{for all }x\geq1\text{ and }t\geq t_0 .
\]
\end{lemma}

\begin{proof}
See \citet[Proposition~0.8(ii)]{resnick1987extreme}.
\end{proof}

\subsection{Heavy-tailed homogeneous structural causal models}\label{defhscm}
We now present the \emph{heavy-tailed homogeneous structural causal models}
(HT-HSCMs). They contain the recursive max-linear models and the sum-linear
models with positive coefficients and nonnegative innovations considered in
the literature.

\begin{definition}\label{Def:HT-HSCM}
An HT-HSCM $\boldsymbol{X}=(X_1,\ldots,X_d)$ on a DAG
$\mathcal{D}=(V,E)$ is specified by the structural assignments
\begin{equation}\label{eq:local-hscm}
    X_i=f_i(\boldsymbol{X}_{pa(i)},\varepsilon_i),
    \qquad i=1,\ldots,d,
\end{equation}
where a fixed ordering is used for each parent vector and the following
conditions hold:
\begin{itemize}
    \item $\varepsilon_1,\ldots,\varepsilon_d$ are i.i.d.\ nonnegative
    regularly varying random variables with tail index $\alpha>0$, namely,
    \begin{equation}\label{eq:RV tail}
          \mathbb{P}(\varepsilon_i>x)\sim \ell(x)x^{-\alpha},
        \quad x \to \infty,
    \end{equation}
    for some slowly varying function $\ell$.
    \item For every $i\in V$, the structural function
    $f_i:\mathbb{R}_+^{|pa(i)|+1}\to\mathbb{R}_+$ satisfies:
    \begin{enumerate}
        \item \textbf{Positive definiteness:}
        $f_i(\boldsymbol{x})=0$ if and only if $\boldsymbol{x}=0$.
        \item \textbf{Continuity:} $f_i$ is continuous.
        \item \textbf{Positive homogeneity of degree one:} for every $c>0$,
        \begin{equation}
            f_i(c\boldsymbol{x})=c f_i(\boldsymbol{x}),
            \qquad \boldsymbol{x}\in\mathbb{R}_+^{|pa(i)|+1}.
        \end{equation}
        \item \textbf{Coordinatewise monotonicity:} if $0\leq x\leq y$
        componentwise, then $f_i(\boldsymbol{x})\leq f_i(\boldsymbol{y})$.
        \item \textbf{Noise normalization:} writing the noise as the final
        coordinate,
        \begin{equation}\label{eq:noise-normalization}
            f_i(0,\ldots,0,1)=1.
        \end{equation}
    \end{enumerate}
    \item The marginal distribution of each $X_i$ is continuous.
\end{itemize}
\end{definition}

\begin{remark}\label{rem:deletion-monotonicity}
Coordinatewise monotonicity implies the coordinate-deletion property used in
the original formulation of \citet{routh2026heavy}. In particular, for
$J\subseteq\{1,\ldots,|pa(i)|+1\}$ and $m_J=\min_{j\in J}x_j$,
\begin{equation}\label{eq:monotone imp}
    f_i(x)\geq f_i(m_J\mathbf{1}_J)
    =m_J f_i(\mathbf{1}_J),
\end{equation}
where $\mathbf{1}_J=\sum_{j\in J}\mathbf{e}_j$. The stronger coordinatewise
condition is needed because coordinate deletion alone is generally not
preserved under composition.
\end{remark}

By recursively substituting the local structural equations along the
ancestral relations, each component $X_i$ of $\boldsymbol{X}$ can be expressed
as a function $F_{{An}(i)}$ of its ancestral noise variables:
\begin{equation}\label{eq:anc agg}
    X_i = F_{{An}(i)} (\boldsymbol{\varepsilon}_{An(i)}), \quad i=1,\dots,d,
\end{equation}
where $\boldsymbol{\varepsilon}_{An(i)}=(\varepsilon_h)_{h\in An(i)}$ is
written in a fixed order. Induction along any topological ordering of the DAG
shows that $F_{An(i)}$ is continuous, positively homogeneous of degree one,
coordinatewise nondecreasing, and vanishes only at the origin.

\begin{definition}\label{Def:AIR}
    The \emph{ancestral impulse response} (AIR) matrix
    $F\in\mathbb{R}_+^{d\times d}$ is defined by
\[
    F_{hi} = \begin{cases}
        F_{{An}(i)}\!\left(\mathbf{e}^{(i)}_h\right)
        & \text{if } h \in {An}(i), \\
        0 & \text{otherwise.}
    \end{cases}
\]
Here $\mathbf{e}^{(i)}_h$ denotes the coordinate vector in
$\mathbb{R}^{|An(i)|}$ corresponding to $h$.
\end{definition}
If $h\in An(i)$, then $F_{hi}$ is the response at node $i$ to a unit
impulse in $\varepsilon_h$, with all other ancestral noises set to zero.
Positive homogeneity implies that an impulse of size $t\geq0$ produces the
response $tF_{hi}$. Moreover, positive definiteness and
\eqref{eq:noise-normalization} give
\begin{equation}\label{eq:AIR-support-diagonal}
    F_{hi}>0 \iff h\in An(i),
    \qquad F_{ii}=1.
\end{equation}
Thus, the support of $F$ is exactly the reflexive reachability relation of the
DAG.

The following lemma describes the limiting joint tail of an HT-HSCM. For
$h\in V$, write
\[
    \boldsymbol F_{h\cdot}:=(F_{h1},\ldots,F_{hd})
\]
for the $h$th row of the AIR matrix, and recall that $F_{hi}=0$ whenever
$h\notin An(i)$.

\begin{lemma}\label{distributionHTHSCM}
Let $\boldsymbol X=(X_1,\ldots,X_d)$ be an HT-HSCM as in
Definition~\ref{Def:HT-HSCM}, and define the Borel measure
$\nu_{\boldsymbol X}$ on $[0,\infty)^d\setminus\{\mathbf 0\}$ by
\begin{equation}\label{eq:nuX}
    \nu_{\boldsymbol X}(A)
    :=\sum_{h=1}^{d}\int_0^\infty
      \mathbf 1_{\{\eta\boldsymbol F_{h\cdot}\in A\}}\,
      \alpha\eta^{-\alpha-1}\,d\eta .
\end{equation}
Then, for every Borel set $A\subseteq[0,\infty)^d\setminus\{\mathbf 0\}$
that is separated from the origin (that is, $A\cap U=\varnothing$ for
some open neighbourhood $U$ of $\mathbf 0$; see
\citealp[Definition~2.1.1]{kulik2020heavy}) and satisfies
$\nu_{\boldsymbol X}(\partial A)=0$,
\begin{equation}\label{eq:HTHSCM-MRV}
    \lim_{x\to\infty}
    \frac{\mathbb P\bigl(x^{-1}\boldsymbol X\in A\bigr)}
         {\mathbb P(\varepsilon_1>x)}
    =\nu_{\boldsymbol X}(A).
\end{equation}
In particular, if $\nu_{\boldsymbol X}(A)>0$, then
$\mathbb P(x^{-1}\boldsymbol X\in A)\sim\nu_{\boldsymbol X}(A)\,
\ell(x)x^{-\alpha}$ as $x\to\infty$.
\end{lemma}

\begin{proof}
Let $\boldsymbol\varepsilon:=(\varepsilon_1,\ldots,\varepsilon_d)$, and let
$\nu_\alpha(d\eta):=\alpha\eta^{-\alpha-1}\,d\eta$ on $(0,\infty)$. Since
$\varepsilon_1,\ldots,\varepsilon_d$ are i.i.d., nonnegative and regularly
varying with index $\alpha$, the vector $\boldsymbol\varepsilon$ is
regularly varying with exponent measure
\[
    \nu_{\boldsymbol\varepsilon}
    =\sum_{h=1}^{d}
     \delta_0\otimes\cdots\otimes\nu_\alpha\otimes\cdots\otimes\delta_0 ,
\]
with $\nu_\alpha$ in the $h$th position and with $\mathbb P(\varepsilon_1>x)$
as normalization \citep[Example~2.1.4]{kulik2020heavy}.

Define $H:[0,\infty)^d\to[0,\infty)^d$ by
\[
    H(\boldsymbol y)
    :=\bigl(F_{An(1)}(\boldsymbol y_{An(1)}),\ldots,
            F_{An(d)}(\boldsymbol y_{An(d)})\bigr),
    \qquad
    \boldsymbol y_{An(i)}:=(y_h)_{h\in An(i)} .
\]
By \eqref{eq:anc agg}, $\boldsymbol X=H(\boldsymbol\varepsilon)$. Since each
$F_{An(i)}$ is continuous and positively homogeneous of degree one, so is
$H$. Moreover, $H(\boldsymbol y)=\mathbf 0$ only if $\boldsymbol y=\mathbf 0$,
because $h\in An(h)$ for every $h$ and each $F_{An(i)}$ vanishes only at the
origin; by continuity and homogeneity, $\|H(\boldsymbol y)\|\geq c\|\boldsymbol y\|$
for some $c>0$, so $H^{-1}(A)$ is separated from the origin whenever $A$ is.
(To regard $H$ as a map on $\mathbb R^d$, set
$H(\boldsymbol y):=H(|y_1|,\ldots,|y_d|)$; this does not change
$H(\boldsymbol\varepsilon)$.) By the transfer of regular variation under
continuous homogeneous maps \citep[Proposition~2.1.12, with
$\gamma=1$]{kulik2020heavy}, $\boldsymbol X=H(\boldsymbol\varepsilon)$ is
regularly varying with exponent measure
$\nu_{\boldsymbol\varepsilon}\circ H^{-1}$ and the same normalization,
provided this measure is not identically zero. That is,
\[
    \lim_{x\to\infty}
    \frac{\mathbb P\bigl(x^{-1}\boldsymbol X\in A\bigr)}
         {\mathbb P(\varepsilon_1>x)}
    =\nu_{\boldsymbol\varepsilon}\circ H^{-1}(A)
\]
for every Borel set $A$ separated from the origin with
$\nu_{\boldsymbol\varepsilon}\circ H^{-1}(\partial A)=0$.

It remains to show that $\nu_{\boldsymbol\varepsilon}\circ H^{-1}
=\nu_{\boldsymbol X}$ and that this measure is not zero. The measure
$\nu_{\boldsymbol\varepsilon}$ is concentrated on the coordinate axes
$\{\eta\mathbf e_h:\eta>0\}$, so positive homogeneity gives, for every Borel
set $A$,
\[
    \nu_{\boldsymbol\varepsilon}\circ H^{-1}(A)
    =\sum_{h=1}^{d}\int_0^\infty
      \mathbf 1_{\{H(\eta\mathbf e_h)\in A\}}\,\alpha\eta^{-\alpha-1}\,d\eta
    =\sum_{h=1}^{d}\int_0^\infty
      \mathbf 1_{\{\eta H(\mathbf e_h)\in A\}}\,\alpha\eta^{-\alpha-1}\,d\eta .
\]
The $i$th coordinate of $H(\mathbf e_h)$ is
$F_{An(i)}(\mathbf e^{(i)}_h)=F_{hi}$ if $h\in An(i)$, and
$F_{An(i)}(\mathbf 0)=0=F_{hi}$ otherwise. Hence
$H(\mathbf e_h)=\boldsymbol F_{h\cdot}$ by Definition~\ref{Def:AIR}, and
$\nu_{\boldsymbol\varepsilon}\circ H^{-1}=\nu_{\boldsymbol X}$. Finally,
$F_{11}=1$, so
$\nu_{\boldsymbol X}(\{\boldsymbol z:z_1>1\})
\geq\int_1^\infty\alpha\eta^{-\alpha-1}\,d\eta=1>0$. This proves
\eqref{eq:HTHSCM-MRV}. The last statement follows from
$\mathbb P(\varepsilon_1>x)\sim\ell(x)x^{-\alpha}$; see \eqref{eq:RV tail}.
\end{proof}

The measure $\nu_{\boldsymbol X}$ is concentrated on the $d$ rays
$\{\eta\boldsymbol F_{h\cdot}:\eta>0\}$, $h\in V$. This is the
single-big-jump principle: an extreme value of $\boldsymbol X$ is caused by
one large innovation $\varepsilon_h$, which the model transmits to every
node in the direction $\boldsymbol F_{h\cdot}$.

\subsection{Examples}

In this subsection, we present several concrete examples of SCMs that fall
within the class of HT-HSCMs introduced in Definition~\ref{Def:HT-HSCM}.
Throughout, we assume that $\varepsilon_1,\ldots,\varepsilon_d$ are i.i.d.\
nonnegative regularly varying random variables satisfying \eqref{eq:RV tail}.
Let $\mathcal{P}_{hi}$ denote the set of all directed paths from node $h$ to
node $i$. For a node with no parents, an empty sum or maximum is understood to
equal zero.

\begin{enumerate}
\item \textbf{Sum-linear SCM}.
Consider the structural equation
\[
    X_i = \sum_{h \in {pa}(i)} c_{hi} X_h + \varepsilon_i,
\]
where the structural function is given by
\[
    f_i(x_{{pa}(i)}, z) := \sum_{h \in {pa}(i)} c_{hi} x_h + z,
\]
and the coefficients $c_{hi}>0$. In this setting, the structural equations
dictate that the noise variables are additively combined. By recursively
composing the structural functions along the DAG, we obtain an induced
ancestral aggregation map $F_{{An}(i)}$ that expands additively over the
network topology. The AIRs in Definition~\ref{Def:AIR} are given by
\begin{equation}\label{eq:linear SCM AIR}
    F_{hi} =
    \begin{cases}
        \sum_{\pi \in \mathcal{P}_{hi}} \prod_{(u, v) \in \pi} c_{uv}, & h \in {an}(i), \\[4pt]
        1, & h = i, \\[4pt]
        0, & h \notin {An}(i).
    \end{cases}
\end{equation}
Here, the product $\prod_{(u,v)\in \pi}c_{uv}$ is taken over all directed
edges along the path $\pi$. The matrix $F$ encodes the total marginal
influence of each ancestor by summing the accumulated weights across all
possible mediating paths.

\item \textbf{Max-linear model}.
Consider the structural equation utilizing the maximum operator, denoted by
the symbol $\bigvee$:
\[
    X_i = \bigvee_{h \in {pa}(i)} (c_{hi} X_h) \vee \varepsilon_i,
\]
where the structural function evaluates to the largest component:
\[
    f_i(x_{{pa}(i)}, z) := \bigvee_{h \in {pa}(i)} (c_{hi} x_h) \vee z,
\]
and the coefficients $c_{hi}>0$. Unlike the linear model, extreme values
propagate exclusively through the highest-weight path rather than
accumulating. The AIRs are given by
\[
    F_{hi} = \begin{cases}
        \bigvee_{\pi \in \mathcal{P}_{hi}} \prod_{(u , v) \in \pi} c_{uv}, & h \in {an}(i), \\[4pt]
        1, & h = i, \\[4pt]
        0, & h \notin {An}(i).
    \end{cases}
\]
In this case, the matrix $F$ isolates the single dominant path transmitting
the extreme noise from $h$ to $i$.

\item \textbf{$\ell_p$-aggregation model.}\label{lpmodel}
Consider the structural equation utilizing $\ell_p$ aggregation:
\begin{equation}
    X_i=f_i\!\left((X_h)_{h\in pa(i)},\varepsilon_i\right),
    \qquad i=1,\ldots,d,
    \label{eq:htlp}
\end{equation}
where the structural function is given by the $\ell_p$ aggregation map
\[
    f_i(x_{{pa}(i)}, z) := \Bigl\| \bigl( (c_{hi} x_h)_{h\in{pa}(i)},\, z \bigr) \Bigr\|_p, \qquad p\in(0,\infty),
\]
and the coefficients $c_{hi}>0$. For $p\geq 1$, this is the usual $\ell_p$
norm; for $0<p<1$, it is an $\ell_p$ quasi-norm. The induced ancestral
aggregation function $F_{{An}(i)}$ in \eqref{eq:anc agg} is nonlinear, and
its AIRs admit the explicit representation
\[
F_{hi}
=
\begin{cases}
\displaystyle
\left[
\sum_{\pi\in\mathcal{P}_{hi}}
\left(
\prod_{(u,v)\in\pi} c_{uv}
\right)^p
\right]^{1/p},
& h\in an(i), \\[12pt]
1,
& h=i, \\[4pt]
0,
& h\notin An(i).
\end{cases}
\]
Thus, the influence transmitted from $h$ to $i$ is obtained by taking the
$\ell_p$-aggregate of the weights accumulated along all directed paths from
$h$ to $i$. We treat the max-linear model as the case $p=\infty$. As
$p\to\infty$, the $\ell_p$ aggregation map converges to the maximum, and the
AIR above converges to the max-linear AIR.
\end{enumerate}

\begin{lemma}[Multiplicative transitivity of AIRs]\label{lem:AIR-transitivity}
For every $h,i,j\in V$, the AIR matrix of an HT-HSCM satisfies
\begin{equation}\label{inequality}
    F_{hj}\geq F_{hi}F_{ij}.
\end{equation}
In particular, taking $i=1$ and $j=2$ gives
$F_{h2}\geq F_{h1}F_{12}$ for every $h\in V$.
\end{lemma}

\begin{proof}
If $i\notin An(j)$, then $F_{ij}=0$, so \eqref{inequality} follows from
$F_{hj}\geq0$. Similarly, if $h\notin An(i)$, then $F_{hi}=0$. It therefore
suffices to consider $h\in An(i)$ and $i\in An(j)$.

We prove, by induction along a topological ordering of the nodes in $De(i)$,
that
\begin{equation}\label{eq:induction-air}
    F_{hv}\geq F_{hi}F_{iv}, \qquad v\in De(i).
\end{equation}
For $v=i$, this is an equality because $F_{ii}=1$. Now let
$v\in de(i)$ and suppose that \eqref{eq:induction-air} holds for every parent
of $v$ that lies in $De(i)$. If $u\in pa(v)\setminus De(i)$, then
$F_{iu}=0$, whereas $F_{hu}\geq0$. Consequently,
\[
    F_{hu}\geq F_{hi}F_{iu}
    \qquad\text{for every }u\in pa(v).
\]
Because $v\neq h$ whenever $F_{hi}>0$ and $v\in de(i)$, the noise coordinate
at $v$ is zero under both impulses. Coordinatewise monotonicity and positive
homogeneity of $f_v$ therefore yield
\begin{align*}
    F_{hv}
    &=f_v\bigl((F_{hu})_{u\in pa(v)},0\bigr)\\
    &\geq f_v\bigl((F_{hi}F_{iu})_{u\in pa(v)},0\bigr)\\
    &=F_{hi}f_v\bigl((F_{iu})_{u\in pa(v)},0\bigr)
      =F_{hi}F_{iv}.
\end{align*}
This completes the induction. Taking $v=j$ proves \eqref{inequality}.
\end{proof}

\begin{remark}\label{rem:unnormalized-air}
The normalization \eqref{eq:noise-normalization} is essential for the form of
Lemma~\ref{lem:AIR-transitivity}. Without it, $F_{ii}$ need not equal one, and
the scale-invariant version of the bound is
\[
    F_{hj}\geq \frac{F_{hi}F_{ij}}{F_{ii}}.
\]
The three model classes above all satisfy $F_{ii}=1$.
\end{remark}

\begin{lemma}\label{lem:ratio-lower-bound}
For all $i,j\in V$, $X_j\geq F_{ij}X_i$ almost surely. In particular,
$X_j/X_i\geq F_{ij}$ almost surely.
\end{lemma}

\begin{proof}
If $i\notin An(j)$, then $F_{ij}=0$, and if $i=j$, then $F_{ii}=1$. Let
$i\in an(j)$. Substituting the structural equations recursively for the
nodes in $An(j)$, but keeping $X_i$ as an input, gives
$X_j=g(X_i,\boldsymbol\varepsilon_K)$ with $K\subseteq An(j)\setminus\{i\}$
and a continuous, coordinatewise nondecreasing, positively homogeneous map
$g$. A unit impulse in $\varepsilon_i$ gives $X_i=F_{ii}=1$ and
$\boldsymbol\varepsilon_K=\mathbf 0$, so $g(1,\mathbf 0)=F_{ij}$. Hence
$X_j\geq g(X_i,\mathbf 0)=X_i\,g(1,\mathbf 0)=F_{ij}X_i$.
\end{proof}

\section{A Consistent Estimator of the AIR}
\label{sec:estimator}

Let
\[
\bigl(X_{1r},\ldots,X_{dr}\bigr), \qquad r=1,\ldots,n,
\]
be i.i.d.\ copies of an HT-HSCM $\boldsymbol{X}=(X_1,\ldots,X_d)$ defined on a
DAG $\mathcal{D}=(V,E)$.

\begin{definition}[Extremal minimum ratio estimator]\label{def:EMRE}
Let $1,2\in V$, and let
\[
X_{1,(1)}\leq \cdots \leq X_{1,(n)}
\]
denote the order statistics of $(X_{1r})_{r=1}^n$. For a sequence
$k_n<n$, define the set of indices corresponding to the $k_n$ largest
observations of $X_1$ by
\[
D_n
:=
\left\{
1\leq r\leq n:
X_{1r}>X_{1,(n-k_n)}
\right\}.
\]
The \emph{extremal minimum ratio estimator} (EMRE) of $F_{12}$ is
\begin{equation}\label{defEMRE}
    \widehat{F}_{12}
    :=
    \bigwedge_{r\in D_n}
    \frac{X_{2r}}{X_{1r}}.
\end{equation}
\end{definition}

The estimator in \eqref{defEMRE} is motivated by the lower endpoint of the
support of the ratio $X_2/X_1$. By Lemma~\ref{lem:ratio-lower-bound}, and
because the limiting conditional law \eqref{tailratiolimit} established in the
proof of Proposition~\ref{prop:EMRE-consistency} has an atom at $F_{12}$,
\[
\inf \operatorname{Supp}\left(\frac{X_2}{X_1}\right)
=
\begin{cases}
F_{12}, & 1\in \operatorname{An}(2),\\
0, & 1\notin \operatorname{An}(2).
\end{cases}
\]
Thus, $F_{12}>0$ if and only if $1\in\operatorname{An}(2)$. The estimator
$\widehat F_{12}$ restricts the minimum ratio to observations for which
$X_1$ is extreme. Under suitable tail assumptions, we show that this
estimator consistently estimates the AIR $F_{12}$ and derive its rate of
convergence.

\begin{assumption}[Lower tail]\label{ass:lower-tail}
There exists $\delta_1\in(0,1)$ such that, for every $i\in V$,
\begin{equation}\label{cdfcontrol}
    G_i(t)=o\bigl(t^{\alpha/(1-\delta_1)}\bigr), \qquad t\downarrow0,
\end{equation}
where $G_i$ denotes the distribution function of $X_i$.
\end{assumption}

\begin{assumption}[Negative moment]\label{ass:neg-moment}
There exists $\delta_2\in(0,1)$ such that, for every $i\in V$,
\begin{equation}\label{expectation control}
    \mathbb E\!\left[X_i^{-(\alpha+\delta_2)}\right]<\infty.
\end{equation}
\end{assumption}

\begin{remark}\label{rem:assumptions}
(i) Both assumptions concern only the innovations. If $s$ is a source node,
then $X_s=\varepsilon_s$. For every $i$, monotonicity and
\eqref{eq:noise-normalization} give $X_i\geq f_i(\mathbf 0,\varepsilon_i)
=\varepsilon_i$. Hence each assumption holds for every $i\in V$ if and only
if it holds with $X_i$ replaced by $\varepsilon_1$.

(ii) Assumption~\ref{ass:lower-tail} implies
Assumption~\ref{ass:neg-moment}. Since
$\mathbb E[\varepsilon_1^{-s}]=s\int_0^\infty t^{-s-1}
\mathbb P(\varepsilon_1\leq t)\,dt$, the moment is finite for every
$s<\alpha/(1-\delta_1)$. So Assumption~\ref{ass:neg-moment} holds for any
$\delta_2<\min\{1,\alpha\delta_1/(1-\delta_1)\}$. We state it separately
because the proofs use it in this form.

(iii) Both assumptions hold for Fr\'echet innovations, for which
$\mathbb P(\varepsilon_1\leq t)=\exp(-t^{-\alpha})$ vanishes faster than any
power of $t$, and for Pareto innovations with a strictly positive lower
endpoint, for which $\mathbb P(\varepsilon_1\leq t)=0$ for small $t$. They also
hold for the absolute value of a Student's $t$ random variable with
$\nu\in(0,1)$ degrees of freedom, for which $\alpha=\nu$ and
$\mathbb P(\varepsilon_1\leq t)\sim ct$ as $t\downarrow0$ for some $c>0$,
provided that $\delta_1,\delta_2<1-\nu$. They fail when $\nu\geq1$, including
the half-Cauchy case $\nu=1$, because then
$\mathbb E[\varepsilon_1^{-(\alpha+\delta_2)}]=\infty$ for every $\delta_2>0$.
\end{remark}

\begin{proposition}\label{prop:EMRE-consistency}
Let $1,2\in V$ be distinct nodes of an HT-HSCM
(Definition~\ref{Def:HT-HSCM}), and suppose that
\[
    k_n\to\infty, \qquad \frac{k_n}{n}\to0,
\]
as $n\to\infty$. Then the following statements hold.

\begin{enumerate}
    \item\label{consistencyofemre}
    The extremal minimum ratio estimator is consistent for $F_{12}$:
    \begin{equation}
        \widehat F_{12} \xrightarrow{\mathbb P} F_{12}.
   \end{equation}

    \item Suppose, in addition, that
    Assumptions~\ref{ass:lower-tail} and~\ref{ass:neg-moment}
    hold. Let $G_1$ be the distribution function of $X_1$, and let
    \[
        a_n:=U_1(n)=G_1^{\leftarrow}\!\left(1-\frac1n\right),
        \qquad
        \Psi_{\theta}(x):=1-e^{-x^{\theta}},\quad x\geq0,
    \]
    where $\Psi_{\theta}$ is the distribution function of a standard Weibull
    random variable with shape parameter $\theta>0$. Under these additional
    assumptions:
    \begin{enumerate}
        \item  If $1\notin\operatorname{An}(2)$, put
        \[
            \mathcal H_{12}:=An(1)\setminus An(2),
            \qquad
            S_1:=\sum_{h\in An(1)}F_{h1}^{\alpha},
        \]
        and define
        \begin{equation}\label{eq:A-exact-constant}
            A
            :=
            \frac{\sum_{h\in\mathcal H_{12}}F_{h1}^{\alpha}}{S_1}\,
            \mathbb E\bigl[X_2^{-\alpha}\bigr].
        \end{equation}
        Then $0<A<\infty$ and, for every $x\geq0$,
        \begin{equation}\label{eq:nonancestor-weibull-limit}
            \lim_{n\to\infty}
            \mathbb P\!\left(
                a_n(\widehat F_{12}-F_{12})\leq x
            \right)
            =1-e^{-Ax^{\alpha}}
            =\Psi_{\alpha}(A^{1/\alpha}x).
        \end{equation}

        \item Suppose that $1\in\operatorname{an}(2)$ and that the
        structural equations are given by an $\ell_p$-aggregation model
        defined in Example~\ref{lpmodel}, with the same $p$ at every node.
        If $0<p<\infty$, put $b_{hi}:=F_{hi}$ and define
        \[
            B_0:=\{h\in An(1):b_{h2}=F_{12}b_{h1}\},
            \qquad
            B_0^c:=An(1)\setminus B_0,
        \]
        \[
            d_h:=b_{h2}^p-F_{12}^pb_{h1}^p,
            \qquad h\in B_0^c,
        \]
        and
        \begin{equation}\label{eq:Y2-definition-proposition}
            Y_2
            :=
            \sum_{h\in B_0^c}d_h\varepsilon_h^p
            +
            \sum_{h\in An(2)\setminus An(1)}b_{h2}^p\varepsilon_h^p.
        \end{equation}
        Let
        \begin{equation}\label{eq:B-exact-constant}
            B
            :=
            \frac{\displaystyle\sum_{h\in B_0}b_{h1}^{\alpha}}
                 {\displaystyle\sum_{h\in An(1)}b_{h1}^{\alpha}}
            \bigl(pF_{12}^{p-1}\bigr)^{\alpha/p}
            \mathbb E\!\left[Y_2^{-\alpha/p}\right].
        \end{equation}
        Then $0<B<\infty$ and, for every $x\geq0$,
        \begin{equation}\label{eq:ancestor-weibull-limit}
            \lim_{n\to\infty}
            \mathbb P\!\left(
                a_n^p(\widehat F_{12}-F_{12})\leq x
            \right)
            =1-e^{-Bx^{\alpha/p}}
            =\Psi_{\alpha/p}(B^{p/\alpha}x).
        \end{equation}
        If $p=\infty$ (the max-linear model), then
        \begin{equation}\label{maxlinearsuperconsistency}
            \mathbb P(\widehat F_{12}=F_{12})\longrightarrow1.
        \end{equation}
    \end{enumerate}
\end{enumerate}
\end{proposition}

\begin{proof}
We first show that
\begin{equation}\label{tailratiolimit}
    \mathcal L\!\left(\left.\frac{X_2}{X_1}\,\right|\,X_1>x\right)
    \xrightarrow{d}
    \frac{
      \sum_{h\in An(1)}F_{h1}^{\alpha}
      \delta_{F_{h2}/F_{h1}}
    }{
      \sum_{h\in An(1)}F_{h1}^{\alpha}
    }.
\end{equation}

Let
\[
    g(x):=\frac{1}{\mathbb P(\varepsilon_1>x)}.
\]
Lemma~\ref{distributionHTHSCM} gives
\begin{equation}\label{eq:biv-mrv-proof}
    g(x)\mathbb P\!\left(\frac1x(X_1,X_2)\in A\right)
    \longrightarrow
    \nu_{12}(A)
    :=
    \sum_{h\in An(1)\cup An(2)}
    \int_0^\infty
       \mathbf 1_{\{r(F_{h1},F_{h2})\in A\}}
       \alpha r^{-\alpha-1}\,dr
\end{equation}
for every Borel set $A\subset\mathbb R_+^2\setminus\{(0,0)\}$ that is
separated from the origin and satisfies
$\nu_{12}(\partial A)=0$. Indeed, apply Lemma~\ref{distributionHTHSCM} to the
cylinder set $\widetilde A:=\{\boldsymbol z\in[0,\infty)^d:(z_1,z_2)\in A\}$.
This set is separated from the origin because $A$ is, its boundary is
$\{\boldsymbol z:(z_1,z_2)\in\partial A\}$, and
$\nu_{\boldsymbol X}(\widetilde A)=\nu_{12}(A)$ for every such $A$ because
$(F_{h1},F_{h2})=(0,0)$ for $h\notin An(1)\cup An(2)$.

For $a>0$, define
\[
    S_a:=\{(z_1,z_2)\in\mathbb R_+^2:z_2\leq az_1,\ z_1>1\}.
\]
Taking the boundary relative to $\mathbb R_+^2$, we have
\[
    \partial S_a
    =\{(1,z):0\leq z\leq a\}
      \cup\{(z,az):z\geq1\}
    \subseteq
    \{(z_1,z_2):z_2=az_1\text{ or }z_1=1\}.
\]
Hence $\nu_{12}(\partial S_a)=0$ whenever
\[
    a\notin
    \left\{\frac{F_{h2}}{F_{h1}}:h\in An(1)\right\}.
\]
By conditional probability,
\[
    \mathbb P\!\left(\left.\frac{X_2}{X_1}\leq a\,\right|\,X_1>x\right)
    =
    \frac{
      g(x)\mathbb P\!\left(x^{-1}(X_1,X_2)\in S_a\right)
    }{
      g(x)\mathbb P(X_1/x>1)
    }.
\]
The numerator and denominator converge, respectively, to
\[
    \sum_{h\in An(1)}
       \mathbf 1_{\{F_{h2}/F_{h1}\leq a\}}F_{h1}^{\alpha}
    \qquad\text{and}\qquad
    \sum_{h\in An(1)}F_{h1}^{\alpha}.
\]
This proves \eqref{tailratiolimit}. In particular, taking
$A=\{\boldsymbol z:z_1>1\}$ in Lemma~\ref{distributionHTHSCM},
\begin{equation}\label{eq:X1-tail-proof}
    \mathbb P(X_1>x)
    \sim
    \left(\sum_{h\in An(1)}F_{h1}^{\alpha}\right)
    \ell(x)x^{-\alpha}.
\end{equation}

We now prove part~\ref{consistencyofemre}. Write
\[
    \widehat F_{12}
    =\bigwedge_{r=1}^n
    \left[
       \frac{X_{2r}}{X_{1r}}
       \mathbf 1_{\{X_{1r}>X_{1,(n-k_n)}\}}
       +\infty\,
       \mathbf 1_{\{X_{1r}\leq X_{1,(n-k_n)}\}}
    \right].
\]
Let
\[
    U_1(t):=G_1^{\leftarrow}\!\left(1-\frac1t\right),
\]
where $G_1$ is the distribution function of $X_1$. By
\eqref{eq:X1-tail-proof} and the standard intermediate-order-statistic result
(see the proof of Theorem~4.2 in \citealp{resnick2007heavy}),
\[
    \frac{X_{1,(n-k_n)}}{U_1(n/k_n)}
    \xrightarrow{\mathbb P}1.
\]

For $\tau\in(0,1)$, denote the event
\[
    \left\{
      \left|
      \frac{X_{1,(n-k_n)}}{U_1(n/k_n)}-1
      \right|<\tau
    \right\}
\]
by $B_{\tau}$. Then
\begin{align*}
    \mathbb P(\widehat F_{12}>F_{12}+\epsilon)
    &\leq \mathbb P(B_{\tau}^c)
       +\mathbb P(\widehat F_{12}>F_{12}+\epsilon,B_{\tau})\\
    &\leq \mathbb P(B_{\tau}^c)
       +\mathbb P\!\left(
          \bigwedge_{r:\,
          X_{1r}>(1+\tau)U_1(n/k_n)}
          \frac{X_{2r}}{X_{1r}}>F_{12}+\epsilon
       \right).
\end{align*}
Choose $\epsilon'\in(0,\epsilon)$ such that $F_{12}+\epsilon'$ is not an
atom of the limiting distribution in \eqref{tailratiolimit}. Since replacing
$\epsilon$ by $\epsilon'$ only enlarges the probability on the right, the
i.i.d.\ structure across observations and $1-z\leq e^{-z}$ give
\begin{align*}
    &\mathbb P(\widehat F_{12}>F_{12}+\epsilon,B_{\tau})\\
    &\quad\leq
    \exp\Bigg\{
       -n\mathbb P\!\left(
          X_1>(1+\tau)U_1(n/k_n)
       \right)\\
    &\hspace{36mm}\times
       \mathbb P\!\left(
          \left.\frac{X_2}{X_1}\leq F_{12}+\epsilon'
          \,\right|\,
          X_1>(1+\tau)U_1(n/k_n)
       \right)
    \Bigg\}.
\end{align*}
Because $F_{11}=1$, the limiting distribution in
\eqref{tailratiolimit} has a strictly positive atom at $F_{12}$. Therefore,
\[
    \liminf_{n\to\infty}
    \mathbb P\!\left(
       \left.\frac{X_2}{X_1}\leq F_{12}+\epsilon'
       \,\right|\,
       X_1>(1+\tau)U_1(n/k_n)
    \right)>0.
\]
Moreover, regular variation gives
\[
    n\mathbb P\!\left(
       X_1>(1+\tau)U_1(n/k_n)
    \right)
    \sim k_n(1+\tau)^{-\alpha}\longrightarrow\infty.
\]
Thus
$\mathbb P(\widehat F_{12}>F_{12}+\epsilon)\to0$. If
$1\in An(2)$, Lemma~\ref{lem:ratio-lower-bound} gives
$X_2/X_1\geq F_{12}$ almost surely. If $1\notin An(2)$, then $F_{12}=0$
and $X_2/X_1\geq0$. Hence
\[
    \mathbb P(\widehat F_{12}<F_{12}-\epsilon)=0
\]
in both cases, proving consistency.

We next derive the asymptotic distributions. Put
\[
    u_n:=U_1(n/k_n),
    \qquad
    a_n:=U_1(n),
\]
and, for $c>0$, define
\[
    \widehat F_{12}^{(c)}
    :=
    \bigwedge_{r:\,X_{1r}>cu_n}\frac{X_{2r}}{X_{1r}}.
\]
On $B_{\tau}$,
\begin{equation}\label{eq:random-threshold-sandwich}
    \widehat F_{12}^{(1-\tau)}
    \leq \widehat F_{12}
    \leq \widehat F_{12}^{(1+\tau)}.
\end{equation}
Consequently, for any deterministic $t_n>0$,
\begin{align}
&\mathbb P\!\left(\widehat F_{12}^{(1-\tau)}>t_n\right)
  -\mathbb P(B_{\tau}^c)
\leq \mathbb P(\widehat F_{12}>t_n)\notag\\
&\hspace{28mm}\leq
\mathbb P\!\left(\widehat F_{12}^{(1+\tau)}>t_n\right)
  +\mathbb P(B_{\tau}^c).
\label{eq:probability-random-threshold-sandwich}
\end{align}
For fixed $c>0$, independence across observations gives
\begin{equation}\label{eq:deterministic-threshold-binomial}
\mathbb P\!\left(\widehat F_{12}^{(c)}>t_n\right)
=
\left[
1-\mathbb P\!\left(
    \frac{X_2}{X_1}\leq t_n,\ X_1>cu_n
\right)
\right]^n.
\end{equation}

We will repeatedly use the following direct small-factor estimate. Put
\[
    \rho_n:=\frac{u_n}{a_n},
    \qquad
    \beta:=\frac{\alpha}{1-\delta_1}>\alpha.
\]
Then, for every $C>0$,
\begin{equation}\label{eq:small-factor-estimate}
    k_n\mathbb P(\varepsilon_1\leq C\rho_n)\longrightarrow0.
\end{equation}
Indeed, choose $\zeta\in(0,1/\alpha-1/\beta)$. Since $U_1\in
RV_{1/\alpha}$ by \eqref{eq:X1-tail-proof} and Section~\ref{sec:prelim},
Lemma~\ref{Potter's bounds}, applied to $f=U_1$ with $\zeta\wedge\frac12$ in
place of $\epsilon$, $t=n/k_n\to\infty$ and $x=k_n\geq1$, gives
$U_1(n)/U_1(n/k_n)\geq\frac12k_n^{1/\alpha-\zeta}$. Hence, for all
sufficiently large $n$,
\[
    \rho_n\leq 2k_n^{-1/\alpha+\zeta}.
\]
Let $s$ be a source node of the DAG. Then $X_s=\varepsilon_s$, so
Assumption~\ref{ass:lower-tail} applied at $i=s$ yields
\[
    \mathbb P(\varepsilon_1\leq t)=o(t^\beta),
    \qquad t\downarrow0.
\]
Consequently,
\[
    k_n\mathbb P(\varepsilon_1\leq C\rho_n)
    =o\!\left(k_n\rho_n^\beta\right)
    =o\!\left(k_n^{1-\beta/\alpha+\beta\zeta}\right)
    \longrightarrow0.
\]
In particular, \eqref{eq:small-factor-estimate} also holds with
$\varepsilon_1$ replaced by any positive random variable $Y$ satisfying
$Y\geq c_0\varepsilon_j$ almost surely for some $c_0>0$ and some $j$.

We now prove part~(a). Let $1\notin An(2)$, so that $F_{12}=0$, and split
the ancestors of node~$1$ into the common and the private ones,
\[
    \mathcal C_{12}:=An(1)\cap An(2),
    \qquad
    \mathcal H_{12}=An(1)\setminus An(2).
\]
The set $\mathcal H_{12}$ is nonempty because it contains node~$1$. Write
$\boldsymbol\varepsilon_{12}:=(\varepsilon_h)_{h\in\mathcal C_{12}}$ and
$\boldsymbol\varepsilon_{12'}:=(\varepsilon_h)_{h\in\mathcal H_{12}}$, let
$M_{12}$ and $M_{12'}$ be their maxima ($M_{12}:=0$ if
$\mathcal C_{12}=\varnothing$), and order the coordinates of $F_{An(1)}$ so
that $X_1=F_{An(1)}(\boldsymbol\varepsilon_{12},\boldsymbol\varepsilon_{12'})$.
Define the \emph{private part} of $X_1$,
\[
    Z_{12'}:=F_{An(1)}\bigl(\mathbf 0,\boldsymbol\varepsilon_{12'}\bigr),
\]
the value that $X_1$ would take if every common-ancestor noise were switched
off. We use the constants
\[
    \overline c_1:=F_{An(1)}(1,\ldots,1),\qquad
    \underline c_{\mathcal C}:=\min_{h\in\mathcal C_{12}}F_{h2},\qquad
    \underline c_{\mathcal H}:=\min_{h\in\mathcal H_{12}}F_{h1},
\]
all positive by \eqref{eq:AIR-support-diagonal}. Monotonicity and
homogeneity give
\begin{equation}\label{eq:nonanc-basic-bounds}
    X_1\leq\overline c_1\,(M_{12}\vee M_{12'}),\qquad
    X_2\geq\underline c_{\mathcal C}\,M_{12},\qquad
    X_2\geq\varepsilon_2,\qquad
    Z_{12'}\leq X_1 .
\end{equation}

The variable $Z_{12'}$ is a function of
$\boldsymbol\varepsilon_{12'}$ only, and
$\mathcal H_{12}\cap An(2)=\varnothing$, so $Z_{12'}$ is independent of
$X_2$. The map $\boldsymbol y\mapsto F_{An(1)}(\mathbf 0,\boldsymbol y)$ is
continuous and positively homogeneous with value $F_{h1}$ at $\mathbf e_h$,
so the argument in the proof of Lemma~\ref{distributionHTHSCM},
with this map in place of $H$,
gives $\mathbb P(Z_{12'}>t)\sim\sum_{h\in\mathcal H_{12}}F_{h1}^{\alpha}\,
\mathbb P(\varepsilon_1>t)$. Together with \eqref{eq:X1-tail-proof} and
$n\mathbb P(X_1>a_n)=1$, this yields $\mathbb P(Z_{12'}>\cdot)\in RV_{-\alpha}$
and
\begin{equation}\label{eq:private-part-tail}
    n\mathbb P\bigl(Z_{12'}>a_nz_n\bigr)\longrightarrow
    \frac{\sum_{h\in\mathcal H_{12}}F_{h1}^{\alpha}}{S_1}\,z^{-\alpha}
    \qquad\text{whenever } z_n\to z\in(0,\infty),
\end{equation}
where the version with $z_n\to z$ follows from monotonicity and continuity
of the limit. The same limit holds with $\geq$ in place of $>$.

By \eqref{eq:probability-random-threshold-sandwich} and
\eqref{eq:deterministic-threshold-binomial} with $t_n=x/a_n$, it suffices to
show that, for every fixed $c>0$ and $x>0$,
\begin{equation}\label{eq:nonanc-target}
    n\mathbb P\!\left(\frac{X_2}{X_1}\leq\frac{x}{a_n},\ X_1>cu_n\right)
    \longrightarrow Ax^{\alpha}.
\end{equation}
Since $X_2/X_1\leq x/a_n$ means $X_1\geq a_nX_2/x$, and $a_nX_2/x\leq cu_n$
exactly when $X_2\leq cx\rho_n$, the left-hand side of
\eqref{eq:nonanc-target} lies between $J_n$ and $I_n+J_n$, where
\[
    I_n:=n\mathbb P\bigl(X_1>cu_n,\ X_2\leq cx\rho_n\bigr),
    \qquad
    J_n:=n\mathbb P\bigl(X_1\geq a_nX_2/x,\ X_2>cx\rho_n\bigr);
\]
on the event in $J_n$, the condition $X_1>cu_n$ holds automatically. Recall
that $\rho_n=u_n/a_n\to0$ and $u_n\to\infty$.

\emph{Step~1: $I_n\to0$.} By \eqref{eq:nonanc-basic-bounds},
\[
    I_n\leq n\mathbb P\bigl(\overline c_1M_{12}>cu_n,\ \underline c_{\mathcal C}M_{12}\leq cx\rho_n\bigr)
       +n\mathbb P\bigl(\overline c_1M_{12'}>cu_n,\ X_2\leq cx\rho_n\bigr).
\]
The first event is empty once $a_n>\overline c_1x/\underline c_{\mathcal C}$,
because it requires
$cu_n/\overline c_1<M_{12}\leq cx\rho_n/\underline c_{\mathcal C}
=cxu_n/(\underline c_{\mathcal C}a_n)$; it is empty for every $n$ if
$\mathcal C_{12}=\varnothing$. For the second term, $M_{12'}$ and $X_2$ are
independent, $n\mathbb P(M_{12'}>cu_n/\overline c_1)\sim
|\mathcal H_{12}|\,(c/\overline c_1)^{-\alpha}k_n/S_1=O(k_n)$, and
$\mathbb P(X_2\leq cx\rho_n)=o(1/k_n)$ by \eqref{eq:small-factor-estimate},
since $X_2\geq\varepsilon_2$. Hence $I_n\to0$.

\emph{Step~2: lower bound for $J_n$.} Fix $\epsilon\in(0,x/2]$. For all large
$n$, $cx\rho_n<\epsilon$, and $Z_{12'}\leq X_1$ gives
\[
    J_n\geq n\mathbb P\bigl(Z_{12'}\geq a_nX_2/x,\ X_2>\epsilon\bigr)
    =\int_{(\epsilon,\infty)}n\mathbb P\bigl(Z_{12'}\geq a_ny/x\bigr)\,F_{X_2}(dy).
\]
The integrand converges pointwise by \eqref{eq:private-part-tail} and is
bounded by $\sup_nn\mathbb P(Z_{12'}\geq a_n\epsilon/x)<\infty$, so dominated
convergence gives
$\liminf_nJ_n\geq\frac{\sum_{\mathcal H_{12}}F_{h1}^{\alpha}}{S_1}\,
x^{\alpha}\,\mathbb E\bigl[X_2^{-\alpha}\mathbf 1_{\{X_2>\epsilon\}}\bigr]$.
Letting $\epsilon\downarrow0$ and using monotone convergence,
$\liminf_nJ_n\geq Ax^{\alpha}$.

\emph{Step~3: comparison of $X_1$ with its private part.}
We show that there are deterministic $\Gamma_n\geq1$ with $\Gamma_n\to1$ such
that, for all large $n$, almost surely,
\begin{equation}\label{eq:nonanc-comparison}
    X_1\leq\Gamma_n\,Z_{12'}\qquad\text{on }\{X_1\geq a_nX_2/x\}.
\end{equation}
If $\mathcal C_{12}=\varnothing$, then $X_1=Z_{12'}$ and $\Gamma_n:=1$.
Otherwise put $\gamma_n:=\overline c_1x/(\underline c_{\mathcal C}a_n)$ and take
$n$ so large that $\gamma_n<1$. Work on $\{X_1\geq a_nX_2/x\}\cap\{X_2>0\}$;
the second event has probability one because $X_2$ has a continuous
distribution. By \eqref{eq:nonanc-basic-bounds},
\[
    \underline c_{\mathcal C}M_{12}\leq X_2\leq\frac{x}{a_n}X_1
    \leq\frac{x\,\overline c_1}{a_n}\,(M_{12}\vee M_{12'}).
\]
If $M_{12}\geq M_{12'}$, this reads
$\underline c_{\mathcal C}M_{12}\leq\underline c_{\mathcal C}\gamma_nM_{12}$,
which forces $M_{12}=M_{12'}=0$, hence $X_1=0$ and $X_2=0$; this is excluded.
Therefore $M_{12'}>M_{12}$, and the same display gives
\begin{equation}\label{eq:nonanc-common-small}
    M_{12}\leq\gamma_n\,M_{12'}.
\end{equation}
Let $\mathcal S:=\{\boldsymbol v\in[0,1]^{|\mathcal H_{12}|}:\max_hv_h=1\}$ and
\[
    \Gamma_n:=\sup_{\boldsymbol v\in\mathcal S}
    \frac{F_{An(1)}(\gamma_n\mathbf 1,\boldsymbol v)}{F_{An(1)}(\mathbf 0,\boldsymbol v)}.
\]
For $\boldsymbol v\in\mathcal S$ some coordinate equals one, so monotonicity
gives $F_{An(1)}(\mathbf 0,\boldsymbol v)\geq\underline c_{\mathcal H}>0$;
hence $\Gamma_n$ is finite, and $\Gamma_n\geq1$ by monotonicity. By
homogeneity, with $\boldsymbol v:=\boldsymbol\varepsilon_{12'}/M_{12'}\in\mathcal S$
and $\boldsymbol w:=\boldsymbol\varepsilon_{12}/M_{12'}\in[0,\gamma_n]^{|\mathcal C_{12}|}$
by \eqref{eq:nonanc-common-small},
\[
    \frac{X_1}{Z_{12'}}
    =\frac{F_{An(1)}(\boldsymbol w,\boldsymbol v)}{F_{An(1)}(\mathbf 0,\boldsymbol v)}
    \leq\frac{F_{An(1)}(\gamma_n\mathbf 1,\boldsymbol v)}{F_{An(1)}(\mathbf 0,\boldsymbol v)}
    \leq\Gamma_n ,
\]
which is \eqref{eq:nonanc-comparison}. It remains to show $\Gamma_n\to1$.
Suppose not. Then there are $\epsilon_0>0$, indices $n_1<n_2<\cdots$, and
$\boldsymbol v_k\in\mathcal S$ with
$F_{An(1)}(\gamma_{n_k}\mathbf 1,\boldsymbol v_k)\geq(1+\epsilon_0)
F_{An(1)}(\mathbf 0,\boldsymbol v_k)$ for all $k$. The sequence
$(\boldsymbol v_k)$ need not converge, but $\mathcal S$ is closed and bounded,
so by the Bolzano--Weierstrass theorem a further subsequence
$\boldsymbol v_{k_j}$ converges to some $\boldsymbol v_*\in\mathcal S$. Along it
$\gamma_{n_{k_j}}\to0$, and continuity of $F_{An(1)}$ on both sides gives
$F_{An(1)}(\mathbf 0,\boldsymbol v_*)\geq(1+\epsilon_0)F_{An(1)}(\mathbf 0,\boldsymbol v_*)$,
that is, $F_{An(1)}(\mathbf 0,\boldsymbol v_*)\leq0$. This contradicts
$F_{An(1)}(\mathbf 0,\boldsymbol v_*)\geq\underline c_{\mathcal H}>0$.

\emph{Step~4: upper bound for $J_n$.} By \eqref{eq:nonanc-comparison} and the
independence of $Z_{12'}$ and $X_2$,
\[
    J_n\leq n\mathbb P\bigl(\Gamma_nZ_{12'}\geq a_nX_2/x,\ X_2>cx\rho_n\bigr)
    =\int_{(cx\rho_n,\infty)}n\mathbb P\Bigl(Z_{12'}\geq\frac{a_ny}{x\Gamma_n}\Bigr)F_{X_2}(dy).
\]
Fix $\epsilon\in(0,x/2]$ and split the integral at $\epsilon$. On
$(\epsilon,\infty)$, the integrand converges pointwise to
$\frac{\sum_{\mathcal H_{12}}F_{h1}^{\alpha}}{S_1}x^{\alpha}y^{-\alpha}$ by
\eqref{eq:private-part-tail}, because $x\Gamma_n\to x$, and it is bounded by
$\sup_nn\mathbb P(Z_{12'}\geq a_n\epsilon/(2x))<\infty$ once $\Gamma_n\leq2$.
Dominated convergence therefore gives the limit
$\frac{\sum_{\mathcal H_{12}}F_{h1}^{\alpha}}{S_1}x^{\alpha}\,
\mathbb E[X_2^{-\alpha}\mathbf 1_{\{X_2>\epsilon\}}]$.

On $(cx\rho_n,\epsilon]$ we need Potter's bounds and
Assumption~\ref{ass:neg-moment}. Since
$\mathbb P(Z_{12'}>\cdot)\in RV_{-\alpha}$, Lemma~\ref{Potter's bounds} with
$\delta_2\wedge\frac12$ in place of $\epsilon$ and $x=t/t'$ gives $t_0$ such that
$\mathbb P(Z_{12'}>t')\leq2\,(t'/t)^{-\alpha-\delta_2}\,\mathbb P(Z_{12'}>t)$
whenever $t_0\leq t'\leq t$. For $y\in(cx\rho_n,\epsilon]$ and $\Gamma_n\leq2$,
put $t':=a_ny/(4x)$. Then $a_ny/(x\Gamma_n)\geq2t'>t'$ and
$cu_n/4\leq t'\leq a_n$, so $t'\geq t_0$ for all large $n$. Hence
\[
    n\mathbb P\Bigl(Z_{12'}\geq\frac{a_ny}{x\Gamma_n}\Bigr)
    \leq n\mathbb P(Z_{12'}>t')
    \leq2\Bigl(\frac{y}{4x}\Bigr)^{-\alpha-\delta_2}n\mathbb P(Z_{12'}>a_n)
    \leq K\,y^{-\alpha-\delta_2}
\]
for a constant $K$ that does not depend on $n$, $y$ or $\epsilon$.
Consequently,
\[
    \limsup_{n\to\infty}J_n
    \leq\frac{\sum_{\mathcal H_{12}}F_{h1}^{\alpha}}{S_1}x^{\alpha}\,
    \mathbb E\bigl[X_2^{-\alpha}\mathbf 1_{\{X_2>\epsilon\}}\bigr]
    +K\,\mathbb E\bigl[X_2^{-\alpha-\delta_2}\mathbf 1_{\{X_2\leq\epsilon\}}\bigr].
\]
By Assumption~\ref{ass:neg-moment}, $\mathbb E[X_2^{-\alpha-\delta_2}]<\infty$, so
the second term tends to $0$ as $\epsilon\downarrow0$, while the first tends
to $Ax^{\alpha}$. Hence $\limsup_nJ_n\leq Ax^{\alpha}$.

Steps~1, 2 and~4 give \eqref{eq:nonanc-target} for every
$c>0$. By \eqref{eq:deterministic-threshold-binomial},
$\mathbb P(a_n\widehat F_{12}^{(c)}>x)\to e^{-Ax^{\alpha}}$, and the limit does
not depend on $c$; \eqref{eq:probability-random-threshold-sandwich} with
$c=1\pm\tau$ then proves \eqref{eq:nonancestor-weibull-limit} for $x>0$. For
$x=0$, note that $\widehat F_{12}>0$ almost surely because $X_2$ has a
continuous distribution. Finally, $A>0$ because $1\in\mathcal H_{12}$, and
$A<\infty$ because
$\mathbb E[X_2^{-\alpha}]\leq1+\mathbb E[X_2^{-\alpha-\delta_2}]<\infty$.

Next suppose that $1\in an(2)$ and that the model is an
$\ell_p$-aggregation model with $0<p<\infty$. Again write
$b_{hi}:=F_{hi}$. Since $An(1)\subseteq An(2)$,
\[
    X_2^p
    =\sum_{h\in An(1)}b_{h2}^p\varepsilon_h^p
     +\sum_{h\in An(2)\setminus An(1)}b_{h2}^p\varepsilon_h^p.
\]
For $x>0$, Taylor expansion at $F_{12}>0$ gives
\begin{equation}\label{eq:finite-p-taylor}
    \left(F_{12}+\frac{x}{a_n^p}\right)^p
    =F_{12}^p+\Delta_n,
    \qquad
    \Delta_n=\frac{x'+o(1)}{a_n^p},
    \quad x':=pF_{12}^{p-1}x.
\end{equation}
Retain $B_0$, $B_0^c$, $d_h$, and $Y_2$ from the proposition, and define
\[
    Y_1:=\sum_{h\in B_0}b_{h1}^p\varepsilon_h^p,
    \qquad
    W_0:=\sum_{h\in B_0^c}b_{h1}^p\varepsilon_h^p.
\]
Lemma~\ref{lem:AIR-transitivity} gives $d_h\geq0$, so $d_h>0$ for
$h\in B_0^c$ by the definition of $B_0$; moreover $1\in B_0$, so $Y_1$ is
nondegenerate. Also, $Y_1$ and $Y_2$ are independent, because they depend on
disjoint sets of innovations. Since $X_2^p-F_{12}^pX_1^p=Y_2$ and
$X_1^p=Y_1+W_0$, \eqref{eq:finite-p-taylor} gives
\begin{equation}\label{eq:ancestor-ratio-equivalence}
\frac{X_2}{X_1}\leq F_{12}+\frac{x}{a_n^p}
\quad\Longleftrightarrow\quad
Y_2-\Delta_nW_0\leq\Delta_nY_1.
\end{equation}
If $B_0^c\neq\varnothing$, then
\[
    W_0\leq K_0Y_2,
    \qquad
    K_0:=\max_{h\in B_0^c}\frac{b_{h1}^p}{d_h}<\infty;
\]
if $B_0^c=\varnothing$, then put $W_0=0$ and $K_0:=0$. Therefore, for every
$\eta\in(0,1)$ and all sufficiently large $n$,
\[
    (1-\eta)Y_2
    \leq Y_2-\Delta_nW_0
    \leq Y_2.
\]
Consequently,
\begin{equation}\label{eq:ancestor-event-sandwich}
\{Y_2\leq\Delta_nY_1\}
\subseteq
\left\{
\frac{X_2}{X_1}\leq F_{12}+\frac{x}{a_n^p}
\right\}
\subseteq
\{(1-\eta)Y_2\leq\Delta_nY_1\}.
\end{equation}

The argument in the proof of
Lemma~\ref{distributionHTHSCM}, applied to the continuous, positively homogeneous
map $\boldsymbol y\mapsto(\sum_{h\in B_0}b_{h1}^p|y_h|^p)^{1/p}$, gives
$\mathbb P(Y_1^{1/p}>t)\sim\sum_{h\in B_0}b_{h1}^{\alpha}\,\mathbb P(\varepsilon_1>t)$.
Together with \eqref{eq:X1-tail-proof} and $n\mathbb P(X_1>a_n)=1$, this gives
\begin{equation}\label{eq:Y1-tail-exact}
    n\mathbb P(Y_1>a_n^pz)
    \longrightarrow
    C_Bz^{-\alpha/p},
    \qquad
    C_B:=
    \frac{\sum_{h\in B_0}b_{h1}^{\alpha}}
         {\sum_{h\in An(1)}b_{h1}^{\alpha}}.
\end{equation}
Furthermore, $2\in An(2)\setminus An(1)$ and $b_{22}=1$, so
$Y_2\geq\varepsilon_2^p>0$ almost surely and, by Assumption~\ref{ass:neg-moment},
\[
    \mathbb E[Y_2^{-(\alpha+\delta_2)/p}]<\infty.
\]
Writing $s_n:=a_n^p\Delta_n\to x'$, conditioning on $Y_2$ gives
\begin{align*}
n\mathbb P(Y_2\leq\Delta_nY_1)
&=
n\int_0^\infty
\mathbb P\!\left(Y_1>a_n^p\frac{y}{s_n}\right)F_{Y_2}(dy).
\end{align*}
By \eqref{eq:Y1-tail-exact}, the integrand converges pointwise to
$C_B(x')^{\alpha/p}y^{-\alpha/p}$. Equivalently, this is Breiman's
lemma \citep[Lemma~1.4.3]{kulik2020heavy} applied to $Y_1^{1/p}$ and
$Y_2^{-1/p}$. The displayed negative-moment condition verifies its moment
hypothesis. Consequently,
\begin{equation}\label{eq:Y1Y2-exact-limit}
\begin{split}
n\mathbb P(Y_2\leq\Delta_nY_1)
&\longrightarrow
C_B(x')^{\alpha/p}\mathbb E[Y_2^{-\alpha/p}]\\
&=Bx^{\alpha/p}.
\end{split}
\end{equation}
We next impose $X_1>cu_n$ directly. For $d\in(0,1]$, let
\[
    E_{n,d}:=\{dY_2\leq\Delta_nY_1\},
    \qquad
    T:=Y_2^{1/p},
    \qquad
    b_{n,d}:=c\left(\frac{2x'}{d}\right)^{1/p}\rho_n.
\]
Since $a_n^p\Delta_n\to x'$, on $E_{n,d}\cap\{T>b_{n,d}\}$ we have,
for all sufficiently large $n$,
\[
    Y_1\geq\frac{dT^p}{\Delta_n}>(cu_n)^p,
\]
and hence $X_1>cu_n$. On the other hand, if $T\leq b_{n,d}$, then
$W_0\leq K_0T^p=o(u_n^p)$. Thus, for all sufficiently large $n$,
$X_1>cu_n$ implies $Y_1>(cu_n)^p/2$. Using the independence of $Y_1$
and $T$, we obtain
\begin{align*}
n\mathbb P(E_{n,d},\ T\leq b_{n,d},\ X_1>cu_n)
&\leq
n\mathbb P\!\left(Y_1>\frac{(cu_n)^p}{2}\right)
\mathbb P(T\leq b_{n,d})\\
&=o(1).
\end{align*}
Indeed, the first probability is $O(k_n/n)$ by
\eqref{eq:Y1-tail-exact}, whereas the second is $o(1/k_n)$ by
\eqref{eq:small-factor-estimate}, since $T\geq\varepsilon_2$.
It remains to evaluate $n\mathbb P(E_{n,d},\ T>b_{n,d})$, which equals
\[
    \int_{(b_{n,d}^p,\infty)}
    n\mathbb P\!\left(Y_1\geq a_n^p\frac{dy}{s_n}\right)F_{Y_2}(dy).
\]
On $(\epsilon,\infty)$ for fixed $\epsilon>0$, dominated convergence applies as
in Step~2 of part~(a). On $(b_{n,d}^p,\epsilon]$, the arguments
$a_n^pdy/s_n$ exceed $(cu_n)^p\to\infty$ for all large $n$, so
Lemma~\ref{Potter's bounds} bounds the integrand by $Ky^{-(\alpha+\delta_2)/p}$
exactly as in Step~4 of part~(a), and the displayed negative-moment condition
makes this bound integrable with an integral that vanishes as
$\epsilon\downarrow0$. Hence, by the same conditioning argument that proves
\eqref{eq:Y1Y2-exact-limit}, with $dY_2$ in place of $Y_2$,
\begin{equation}\label{eq:ancestor-direct-threshold-limit}
    n\mathbb P(E_{n,d},\ X_1>cu_n)
    \longrightarrow d^{-\alpha/p}Bx^{\alpha/p}.
\end{equation}

Applying \eqref{eq:ancestor-direct-threshold-limit} with $d=1$ and
$d=1-\eta$ in \eqref{eq:ancestor-event-sandwich}, and then letting
$\eta\downarrow0$, gives, for every fixed $c>0$,
\begin{equation}\label{eq:ancestor-single-observation-exact}
\begin{split}
&n\mathbb P\!\left(
    \frac{X_2}{X_1}\leq F_{12}+\frac{x}{a_n^p},\ X_1>cu_n
\right)\\
&\hspace{15mm}\longrightarrow Bx^{\alpha/p}.
\end{split}
\end{equation}
Combining \eqref{eq:ancestor-single-observation-exact} with
\eqref{eq:deterministic-threshold-binomial} gives
\[
\mathbb P\!\left(
    a_n^p(\widehat F_{12}^{(c)}-F_{12})>x
\right)
\longrightarrow e^{-Bx^{\alpha/p}}.
\]
Finally, \eqref{eq:probability-random-threshold-sandwich}, applied with
$t_n=F_{12}+x/a_n^p$, proves \eqref{eq:ancestor-weibull-limit} for $x>0$.
For $x=0$, both sides vanish: $X_2^p-F_{12}^pX_1^p=Y_2\geq\varepsilon_2^p>0$
almost surely, so $X_2/X_1>F_{12}$ almost surely and
$\mathbb P(\widehat F_{12}=F_{12})=0$. Finally, $B>0$ because $1\in B_0$, and
$B<\infty$ because
$\mathbb E[Y_2^{-\alpha/p}]\leq1+\mathbb E[Y_2^{-(\alpha+\delta_2)/p}]<\infty$.

Finally, consider the case \(p=\infty\), corresponding to a recursive
max-linear model. Suppose that \(1\in an(2)\), and write
\(b_{hi}=F_{hi}\). In particular, \(b_{12}=F_{12}\). Since
\(An(1)\subseteq An(2)\),
\[
X_2
=
\bigvee_{h\in An(2)}b_{h2}\varepsilon_h
=
\left(
\bigvee_{h\in An(1)}b_{h2}\varepsilon_h
\right)
\vee
\left(
\bigvee_{h\in An(2)\setminus An(1)}
b_{h2}\varepsilon_h
\right).
\]

We show that
\[
\mathbb P\left(\widehat F_{12}\neq F_{12}\right)
\longrightarrow 0.
\]
As in the consistency proof, let
\[
u_n^+=(1+\tau)U_1\left(\frac{n}{k_n}\right).
\]
Since \(\mathbb P(B_{\tau}^c)\to0\), it suffices to prove that
\[
n\mathbb P\left(
\frac{X_2}{X_1}=F_{12},\,
X_1>u_n^+
\right)
\longrightarrow\infty.
\]

Observe that
\begin{align*}
&\mathbb P\left(
\frac{X_2}{X_1}=F_{12},\,
X_1>u_n^+
\right)\\
&\qquad\geq
\mathbb P\left(
\frac{X_2}{X_1}=F_{12},\,
X_1>u_n^+,\,
X_1=\varepsilon_1,\,
X_2=b_{12}\varepsilon_1
\right).
\end{align*}
Define
\[
Z
:=
\left(
\bigvee_{h\in An(2)\setminus\{1\}}
\frac{b_{h2}}{b_{12}}\varepsilon_h
\right)
\vee
\left(
\bigvee_{h\in an(1)}
b_{h1}\varepsilon_h
\right).
\]
Since \(Z\) does not involve \(\varepsilon_1\), it is independent of
\(\varepsilon_1\). Moreover,
\[
\{\varepsilon_1>u_n^+,\ \varepsilon_1>Z\}
\subseteq
\left\{
\frac{X_2}{X_1}=F_{12},\,
X_1>u_n^+,\,
X_1=\varepsilon_1,\,
X_2=b_{12}\varepsilon_1
\right\}.
\]
Indeed, \(\varepsilon_1>Z\) implies
\[
\varepsilon_1>
\bigvee_{h\in an(1)}b_{h1}\varepsilon_h,
\]
and therefore \(X_1=\varepsilon_1\). It also implies
\[
b_{12}\varepsilon_1>
\bigvee_{h\in An(2)\setminus\{1\}}b_{h2}\varepsilon_h,
\]
and hence \(X_2=b_{12}\varepsilon_1\). Consequently,
\[
\frac{X_2}{X_1}=b_{12}=F_{12}.
\]

It follows that
\begin{align*}
n\mathbb P\left(
\frac{X_2}{X_1}=F_{12},\,
X_1>u_n^+
\right)
&\geq
n\mathbb P\left(
\varepsilon_1>\max\{u_n^+,Z\}
\right)\\
&\geq
n\mathbb P(\varepsilon_1>u_n^+)
   \mathbb P(Z\leq u_n^+),
\end{align*}
where the last inequality follows from the independence of
\(\varepsilon_1\) and \(Z\).

Let
\[
C_1:=\sum_{h\in An(1)}b_{h1}^{\alpha}.
\]
Since, by \eqref{eq:X1-tail-proof},
\[
\mathbb P(X_1>x)
\sim C_1\mathbb P(\varepsilon_1>x),
\]
regular variation gives
\[
n\mathbb P(\varepsilon_1>u_n^+)
\sim
\frac{(1+\tau)^{-\alpha}}{C_1}\,k_n
\longrightarrow\infty.
\]
Furthermore, \(u_n^+\to\infty\), and \(Z<\infty\) almost surely, so
\[
\mathbb P(Z\leq u_n^+)\longrightarrow1.
\]
Therefore,
\[
n\mathbb P\left(
\frac{X_2}{X_1}=F_{12},\,
X_1>u_n^+
\right)
\longrightarrow\infty.
\]

Finally, on \(B_{\tau}\), every observation satisfying \(X_1>u_n^+\)
belongs to \(D_n\). Since \(X_2/X_1\geq F_{12}\) almost surely
(Lemma~\ref{lem:ratio-lower-bound}),
\begin{align*}
\mathbb P(\widehat F_{12}\neq F_{12})
&\leq
\mathbb P(B_{\tau}^c)
+
\left[
1-
\mathbb P\left(
\frac{X_2}{X_1}=F_{12},\,
X_1>u_n^+
\right)
\right]^n\\
&\leq
\mathbb P(B_{\tau}^c)
+
\exp\left\{
-n\mathbb P\left(
\frac{X_2}{X_1}=F_{12},\,
X_1>u_n^+
\right)
\right\}\\
&\longrightarrow0.
\end{align*}
Thus,
\[
\mathbb P\left(\widehat F_{12}=F_{12}\right)\longrightarrow1.
\]
This part of the argument uses only $k_n\to\infty$ and $k_n/n\to0$, and not
Assumptions~\ref{ass:lower-tail} and~\ref{ass:neg-moment}.
\end{proof}

\begin{remark}\label{rem:rates}
    Suppose that
    $\mathbb P(X_1>x)\sim c_1x^{-\alpha}$, so that $a_n\sim(c_1n)^{1/\alpha}$.
    Proposition~\ref{prop:EMRE-consistency} then gives three regimes. If
    $1\notin An(2)$, then $\widehat F_{12}$ is of exact order $n^{-1/\alpha}$.
    If $1\in an(2)$ in an $\ell_p$-aggregation model with $p<\infty$, then
    $\widehat F_{12}-F_{12}$ is of exact order $n^{-p/\alpha}$. In the
    max-linear model, $\widehat F_{12}=F_{12}$ with probability tending to
    one. In particular, when $0<p<1$ the ancestral error decays more slowly
    than a non-ancestral estimate, and when $p>1$ it decays faster. None of
    these limit laws depends on the choice of $k_n$, provided that
    $k_n\to\infty$ and $k_n/n\to0$.
\end{remark}

\section{Structure Learning}
\label{sec:structure}

\subsection{Causal Order}
Given a causal DAG with node set $V = \{1, \dots, d\}$, a causal order (or
topological order) is a permutation $\pi: V \to V$ satisfying
$u \in an(v) \implies \pi(u) < \pi(v)$. For a causal DAG, there exists at
least one causal order, which may not be unique. Even though a causal order
does not fully identify a DAG, it provides crucial information on causal
relations and reduces the search space for further DAG discovery.

We use the EASE algorithm of \citet[Section~3.1]{gnecco2021causal}, with their
causal tail coefficient \citep[Definition~1]{gnecco2021causal} replaced by
$\widehat F_{ij}$. At the last step, where $V_d$ is a singleton, a maximum
over the empty set is taken to be $0$. See Algorithm~\ref{alg:ease} for more
details.

\renewcommand{\algorithmicrequire}{\textbf{Input:}}
\renewcommand{\algorithmicensure}{\textbf{Output:}}

\begin{algorithm}
\caption{EASE algorithm with estimated AIRs}
\label{alg:ease}
\begin{algorithmic}[1]
\Require
$(\hat{F}_{uv})_{u,v \in V, u \neq v}$ associated with node set $V = \{1, \dots, d\}$
\Ensure Causal order $\pi : V \to V$
\State $V_1 \leftarrow V$
\For{$s = 1$ \textbf{to} $d$}
    \For{\textbf{all} $v \in V_s$}
        \State $M_v^{(s)} \leftarrow \max_{u \in V_s \setminus \{v\}} \hat{F}_{uv}$
    \EndFor
    \State $v_s \in \arg \min_{v \in V_s} M_v^{(s)}$
    \State $\pi(v_s) \leftarrow s$
    \State $V_{s+1} \leftarrow V_s \setminus \{v_s\}$
\EndFor
\State \textbf{return} permutation $\pi$
\end{algorithmic}
\end{algorithm}

\begin{proposition}\label{prop:causal_order_consistency}
Under the assumptions of Part~\ref{consistencyofemre} of
Proposition~\ref{prop:EMRE-consistency}, let $\hat{\pi}$ denote the causal
order estimated by providing $(\hat{F}_{ij})_{i,j \in V, i \neq j}$ as input
to Algorithm~\ref{alg:ease}, and let $\Pi$ be the set of all valid causal
orders for the underlying DAG $\mathcal{D}$. Then, as $n \to \infty$,
\begin{equation}
    \mathbb{P}( \hat{\pi} \notin \Pi) \leq |V|^2 \bigvee_{\substack{i,j \in V \\ i \neq j}} \mathbb{P}\left(|\hat{F}_{ij} - F_{ij}| \geq \frac{m_r}{2}\right) \to 0,
\end{equation}
where $m_r$ is the smallest AIR of an ancestral pair,
\begin{equation}
    m_r = \bigwedge_{\substack{i,j \in V \\ i \in {an}(j)}} F_{ij} > 0,
\end{equation}
with the convention $m_r:=\infty$ if $\mathcal D$ has no edges, in which case
every permutation is a valid causal order.
\end{proposition}

\begin{proof}
The argument follows the proof of Proposition~2 in \citet{gnecco2021causal}.
Let $\mathcal A_n$ be the event that $|\widehat F_{ij}-F_{ij}|<m_r/2$ for all
$i\neq j$, and consider any step $s$ of Algorithm~\ref{alg:ease} on
$\mathcal A_n$. If $v\in V_s$ has no ancestor in $V_s$, then $F_{uv}=0$ for
every $u\in V_s\setminus\{v\}$, and hence $M_v^{(s)}<m_r/2$. If $v$ has an
ancestor $u\in V_s$, then
$M_v^{(s)}\geq\widehat F_{uv}>F_{uv}-m_r/2\geq m_r/2$. Since $V_s$ always
contains a node with no ancestor in $V_s$, the selected node $v_s$ has no
ancestor among the remaining nodes, so every ancestor of $v_s$ has been
selected earlier, and $\widehat\pi\in\Pi$. By the union bound,
$\mathbb P(\widehat\pi\notin\Pi)\leq\mathbb P(\mathcal A_n^c)
\leq d(d-1)\max_{i\neq j}\mathbb P(|\widehat F_{ij}-F_{ij}|\geq m_r/2)$, which
tends to zero by Part~\ref{consistencyofemre} of
Proposition~\ref{prop:EMRE-consistency}.
\end{proof}

\subsection{Consistent Recovery of the Reachability Matrix}
\label{sec:reachability}

Let $R=(R_{ij})_{i,j\in V}$ denote the reachability matrix of $\mathcal D$,
where
\[
R_{ij}:=\mathbf 1_{\{i\in\operatorname{An}(j)\}}.
\]
The diagonal is known, $R_{ii}=1$, so only the off-diagonal entries have to be
estimated.

A valid causal order reduces this task to a family of one-sided decisions. Let
$\pi\in\Pi$ and let $i\neq j$. If $\pi(i)>\pi(j)$, then $i\notin an(j)$,
because every valid causal order places ancestors before their descendants;
hence $R_{ij}=0$ and nothing needs to be tested. If $\pi(i)<\pi(j)$, then
$j\notin an(i)$, so either $i\in an(j)$ or the two nodes are not ancestrally
related, and it remains to decide which. Proposition~\ref{prop:EMRE-consistency}
shows that the EMRE separates these two cases at different scales. If
$i\in an(j)$, then $\widehat F_{ij}\ConvP F_{ij}>0$. If $i\notin An(j)$, then
$F_{ij}=0$ and, by \eqref{eq:nonancestor-weibull-limit} with $i$ as the
conditioning node, $U_i(n)\widehat F_{ij}$ converges in distribution, so that
$\widehat F_{ij}=O_p\bigl(U_i(n)^{-1}\bigr)$ with $U_i(n)\to\infty$.
Consequently, any threshold $\tau_n>0$, deterministic or random, that satisfies
\begin{equation}\label{eq:threshold-conditions}
    \tau_n\ConvP0
    \qquad\text{and}\qquad
    U_i(n)\,\tau_n\ConvP\infty
\end{equation}
gives a consistent decision: declare $i\in an(j)$ if $\widehat F_{ij}>\tau_n$,
and $i\notin an(j)$ otherwise. Indeed, if $i\in an(j)$, the first condition
gives $\mathbb P(\widehat F_{ij}\leq\tau_n)\to0$; if $i\notin An(j)$, the
second condition and the tightness of $U_i(n)\widehat F_{ij}$ give
\[
    \mathbb P\bigl(\widehat F_{ij}>\tau_n\bigr)
    =\mathbb P\bigl(U_i(n)\widehat F_{ij}>U_i(n)\tau_n\bigr)\longrightarrow0.
\]
In words, the threshold must vanish, so that no ancestral relation is lost,
but more slowly than the non-ancestral EMRE, which decays at the rate
$U_i(n)^{-1}$.

The conditions \eqref{eq:threshold-conditions} cannot be imposed directly,
because $U_i(n)$ depends on the unknown tail of $X_i$. For instance, if
$\mathbb P(X_i>x)\sim cx^{-\alpha}$, the deterministic choice
$\tau_n=n^{-\gamma}$ satisfies them if and only if $0<\gamma<1/\alpha$, which
requires knowledge of the tail index. We instead use the data-driven threshold
\[
    \tau_{n,i}:=X_{i,(n-k_n)}^{-1},
\]
the reciprocal of the $(k_n+1)$st largest observation of $X_i$, which
satisfies \eqref{eq:threshold-conditions} automatically. By the
intermediate-order-statistic convergence
$X_{i,(n-k_n)}/U_i(n/k_n)\ConvP1$, the threshold behaves like
$U_i(n/k_n)^{-1}\to0$. Moreover,
\[
    U_i(n)\,\tau_{n,i}
    =\frac{U_i(n)}{U_i(n/k_n)}\cdot\frac{U_i(n/k_n)}{X_{i,(n-k_n)}},
\]
where the second factor converges in probability to one and, by
Lemma~\ref{Potter's bounds}, the first is at least
$(1-\epsilon)\,k_n^{1/\alpha-\epsilon}$ for every $\epsilon\in(0,1/\alpha)$
and all large $n$, so that it tends to infinity. Thus
$X_{i,(n-k_n)}$ estimates a quantile of $X_i$ that is high, but lower than
$U_i(n)$ by a factor growing roughly like $k_n^{1/\alpha}$, and the threshold
adapts to the tail of $X_i$ without estimating it. The resulting decision is
also unchanged if $X_i$ is multiplied by a positive constant, since
$\widehat F_{ij}$ and $\tau_{n,i}$ are then divided by the same constant.

Accordingly, given the causal order $\widehat\pi$ returned by
Algorithm~\ref{alg:ease}, we define
\begin{equation}\label{eq:ordered-reachability-estimator}
\widehat R_{ij}
:=
\begin{cases}
1,
& i=j,\\[1mm]
\mathbf 1_{\left\{
\widehat F_{ij}>X_{i,(n-k_n)}^{-1}
\right\}},
& \widehat\pi(i)<\widehat\pi(j),\\[1mm]
0,
& \widehat\pi(i)>\widehat\pi(j).
\end{cases}
\end{equation}
For example, if the estimated causal order is $(1,2,3,4)$, then we test the
pairs $(1,2),(1,3),(1,4)$, followed by $(2,3),(2,4)$ and $(3,4)$; all pairs in
the reverse direction are set equal to zero.

\begin{proposition}[Consistent recovery of the reachability matrix]
\label{prop:reachability-consistency}
Under the assumptions of Proposition~\ref{prop:EMRE-consistency}, including
Assumptions~\ref{ass:lower-tail} and~\ref{ass:neg-moment}, let
$\widehat\pi$ be the causal order returned by Algorithm~\ref{alg:ease}, and
define $\widehat R$ by \eqref{eq:ordered-reachability-estimator}. Then
\[
\mathbb P\!\left(\widehat R=R\right)\longrightarrow1.
\]
\end{proposition}

\begin{proof}
Let $\Pi$ denote the set of valid causal orders of $\mathcal D$ and put
\[
\mathcal E_n:=\{\widehat\pi\in\Pi\}.
\]
By Proposition~\ref{prop:causal_order_consistency},
$\mathbb P(\mathcal E_n)\to1$.

For each $i\in V$, define
\[
U_i(t):=G_i^{\leftarrow}\!\left(1-\frac{1}{t}\right),
\]
where $G_i$ is the distribution function of $X_i$. The
intermediate-order-statistic convergence for regularly varying tails (see the
argument in the proof of Proposition~\ref{prop:EMRE-consistency}, with $X_i$ in
place of $X_1$) gives
\begin{equation}\label{eq:intermediate-order-statistic}
\frac{X_{i,(n-k_n)}}{U_i(n/k_n)}
\xrightarrow{\mathbb P}1.
\end{equation}

The diagonal entries are equal to one by definition, so let $i\neq j$.
Suppose first that $i\in an(j)$. On $\mathcal E_n$, every valid
causal order places $i$ before $j$, so $\widehat\pi(i)<\widehat\pi(j)$.
Moreover, $F_{ij}>0$ and Proposition~\ref{prop:EMRE-consistency} yields
\[
\widehat F_{ij}\xrightarrow{\mathbb P}F_{ij}.
\]
Since \eqref{eq:intermediate-order-statistic} implies
$X_{i,(n-k_n)}^{-1}\xrightarrow{\mathbb P}0$, we obtain
\[
\mathbb P\!\left(
\widehat F_{ij}>X_{i,(n-k_n)}^{-1}
\right)\longrightarrow1.
\]
Consequently,
\[
\mathbb P(\widehat R_{ij}=0)
\leq
\mathbb P(\mathcal E_n^c)
+
\mathbb P\!\left(
\widehat F_{ij}\leq X_{i,(n-k_n)}^{-1}
\right)
\longrightarrow0.
\]

Now suppose that $i\notin\operatorname{An}(j)$. If
$\widehat\pi(i)>\widehat\pi(j)$, then \eqref{eq:ordered-reachability-estimator}
sets $\widehat R_{ij}=0$. If $\widehat\pi(i)<\widehat\pi(j)$, then the pair is
tested. By \eqref{eq:nonancestor-weibull-limit} in
Proposition~\ref{prop:EMRE-consistency},
\[
\widehat F_{ij}=O_p\!\left(U_i(n)^{-1}\right).
\]
Therefore,
\begin{align*}
X_{i,(n-k_n)}\widehat F_{ij}
&=
\frac{X_{i,(n-k_n)}}{U_i(n/k_n)}
\frac{U_i(n/k_n)}{U_i(n)}
\bigl\{U_i(n)\widehat F_{ij}\bigr\}.
\end{align*}
The first factor converges in probability to one by
\eqref{eq:intermediate-order-statistic}, and the third factor is $O_p(1)$.
Since $X_i\in RV_+(\alpha)$ by Lemma~\ref{distributionHTHSCM} (take
$A=\{\boldsymbol z:z_i>1\}$), its upper quantile function $U_i$ belongs to
$RV_{1/\alpha}$; see Section~\ref{sec:prelim}. Lemma~\ref{Potter's bounds},
applied to $f=U_i$ with $t=n/k_n\to\infty$ and $x=k_n\geq1$, gives, for every
$\epsilon\in(0,1/\alpha)$ and all sufficiently large $n$,
\[
\frac{U_i(n)}{U_i(n/k_n)}=\frac{U_i\bigl(k_n\cdot n/k_n\bigr)}{U_i(n/k_n)}
\geq(1-\epsilon)\,k_n^{1/\alpha-\epsilon}.
\]
Since $k_n\to\infty$, this gives
\[
\frac{U_i(n/k_n)}{U_i(n)} \to 0.
\]
It follows that
\[
X_{i,(n-k_n)}\widehat F_{ij}=o_p(1).
\]
Hence
\begin{align*}
\mathbb P(\widehat R_{ij}=1)
&\leq
\mathbb P\!\left(
\widehat F_{ij}>X_{i,(n-k_n)}^{-1}
\right)\\
&=
\mathbb P\!\left(
X_{i,(n-k_n)}\widehat F_{ij}>1
\right)
\longrightarrow0.
\end{align*}
Thus $\widehat R_{ij}\xrightarrow{\mathbb P}R_{ij}$ for every $i,j\in V$.
Since $V$ is finite, a union bound over all pairs completes the proof.
\end{proof}

\begin{remark}\label{rem:threshold}
Since the argument uses only \eqref{eq:threshold-conditions}, the EMRE and the
threshold may be computed with different intermediate sequences, say $k_n$ in
$\widehat F_{ij}$ and $k_n'$ in $X_{i,(n-k_n')}$, provided both tend to
infinity and are $o(n)$. This is the setting of Sections~\ref{sec:simulation}
and~\ref{sec:danube-sensitivity}.
The choice of $k_n'$ trades off the two types of error. A smaller $k_n'$
gives a smaller threshold and risks declaring non-ancestral pairs ancestral; a
larger $k_n'$ gives a larger threshold and risks missing weak ancestral
relations.
\end{remark}

\section{Simulation Study}
\label{sec:simulation}

We conducted a Monte Carlo study to evaluate the finite-sample performance
of the proposed structure-learning procedure. The procedure first estimates
a causal order using EASE and then estimates the reachability matrix using
the extremal minimum ratio estimator.

\subsection{Simulation Design}

We considered
\[
(d,n)\in\{(5,1000),(10,5000),(30,50000)\},
\qquad
\alpha\in\{1,2\},
\]
with \(50\) independent replications for each configuration.

In each replication, we sampled a causal order uniformly from all
permutations of the \(d\) nodes and selected, uniformly without replacement,
about \(30\%\) of the \(\binom{d}{2}\) admissible forward edges, namely
\(3\), \(14\) and \(130\) edges for \(d=5\), \(10\) and \(30\). Nonzero edge coefficients were
independently generated from
\[
c_{ji}\sim\operatorname{Unif}(0.2,0.9).
\]

Each node independently received an aggregation index
\[
\rho_i\in\{1,2,\infty\},
\qquad
\mathbb P(\rho_i=r)=\frac13.
\]
The structural equations were
\[
X_i=
\begin{cases}
\displaystyle
\left[
\varepsilon_i^{\rho_i}
+
\sum_{j\in\operatorname{pa}(i)}
(c_{ji}X_j)^{\rho_i}
\right]^{1/\rho_i},
& \rho_i\in\{1,2\},\\[3mm]
\displaystyle
\varepsilon_i
\vee
\bigvee_{j\in\operatorname{pa}(i)}
c_{ji}X_j,
& \rho_i=\infty.
\end{cases}
\]
Thus, one simulated DAG could contain sum-linear, Euclidean-norm, and
max-linear mechanisms. The resulting model is itself an HT-HSCM.
Part~\ref{consistencyofemre} and part~(a) of
Proposition~\ref{prop:EMRE-consistency}, together with
Propositions~\ref{prop:causal_order_consistency}
and~\ref{prop:reachability-consistency}, therefore apply to it. Part~(b) does
not apply, because it assumes a common $p$.
The innovations were independent across nodes and
observations and satisfied
\[
\mathbb P(\varepsilon_i>x)=x^{-\alpha},
\qquad x\geq1.
\]
Because $\varepsilon_i\geq1$, these innovations satisfy
Assumptions~\ref{ass:lower-tail} and~\ref{ass:neg-moment} trivially.

The true reachability matrix was obtained by transitive closure:
\[
R_{ij}
=
\mathbf 1_{\{i\in\operatorname{An}(j)\}}.
\]

\subsection{Estimation and Performance Measures}

We considered tuning exponents \(\kappa,\kappa'\) taking values in the grid
\[
\mathcal K
=
\{0.2,0.3,\ldots,0.9\}.
\]
For \(\kappa\in\mathcal K\), let \(k_n=\lfloor n^{\kappa}\rfloor\) and compute
\[
\widehat F_{ij}(\kappa)
=
\min_{\ell:\,
X_{i\ell}>X_{i,(n-k_n)}}
\frac{X_{j\ell}}{X_{i\ell}}.
\]
The matrix \(\widehat F(\kappa)\) was supplied to EASE
(Algorithm~\ref{alg:ease}), which returns an estimated causal order
\(\widehat\pi_{\kappa}\). As in Section~\ref{sec:structure},
\(\widehat\pi_{\kappa}(i)\) denotes the position of node \(i\) in this order.
For \(\kappa'\in\mathcal K\), set \(k_n'=\lfloor n^{\kappa'}\rfloor\) and
define, as in \eqref{eq:ordered-reachability-estimator},
\[
\widehat R_{ij}(\kappa,\kappa')
=
\begin{cases}
1,
& i=j,\\[1mm]
\mathbf 1\left\{
\widehat F_{ij}(\kappa)>
X_{i,(n-k_n')}^{-1}
\right\},
&
\widehat\pi_{\kappa}(i)<\widehat\pi_{\kappa}(j),\\[1mm]
0,
&
\widehat\pi_{\kappa}(i)>\widehat\pi_{\kappa}(j).
\end{cases}
\]
All \(64\) combinations \((\kappa,\kappa')\in\mathcal K^2\) were evaluated
on each simulated data set.

Performance was evaluated using the off-diagonal entries of the
reachability matrix.Let FP,TP,FN denote the false positives, true positives and false negatives respectively. We report the normalized structural Hamming
distance
\[
\operatorname{NSHD}
=
\frac{\operatorname{FP}+\operatorname{FN}}{d(d-1)}
\]
and
\[
\operatorname{Power}
=
\frac{\operatorname{TP}}
{\operatorname{TP}+\operatorname{FN}}.
\]
The ancestral violation rate of the estimated order is
\[
\operatorname{AVR}(\widehat\pi_{\kappa})
=
\frac{
\displaystyle
\sum_{i\neq j}
R_{ij}\,
\mathbf 1\left\{
\widehat\pi_{\kappa}(i)>\widehat\pi_{\kappa}(j)
\right\}
}{
\displaystyle
\sum_{i\neq j}R_{ij}
},
\]
the proportion of ancestral pairs that the estimated order places in the
wrong direction. Lower NSHD and AVR and higher power indicate better
performance.

\subsection{Results}
\begin{table}[!t]
\centering
\small
\setlength{\tabcolsep}{6pt}
\renewcommand{\arraystretch}{1.15}

\caption{Representative reachability-recovery results averaged over
50 Monte Carlo replications.}
\label{tab:reachability-summary}

\begin{tabular}{cccccc}
\toprule
\((d,n)\) & \(\alpha\) & \((\kappa,\kappa')\)
& NSHD & Power & AVR \\
\midrule
\((5,1000)\)
& \(1\) & \((0.4,0.5)\)
& 0.0000 & 1.0000 & 0.0000 \\

\((5,1000)\)
& \(2\) & \((0.4,0.5)\)
& 0.0060 & 0.9760 & 0.0000 \\
\midrule

\((10,5000)\)
& \(1\) & \((0.4,0.5)\)
& 0.0016 & 1.0000 & 0.0000 \\

\((10,5000)\)
& \(2\) & \((0.4,0.5)\)
& 0.0180 & 0.9714 & 0.0000 \\
\midrule

\((30,50000)\)
& \(1\) & \((0.5,0.7)\)
& 0.0027 & 0.9993 & 0.0000 \\

\((30,50000)\)
& \(2\) & \((0.5,0.7)\)
& 0.0307 & 0.9686 & 0.0001 \\
\bottomrule
\end{tabular}
\end{table}

Table~\ref{tab:reachability-summary} reports one tuning pair per dimension,
taken from the low-error region of the tuning grid. The same pair is used for
both values of \(\alpha\), so that the two tail indices can be compared
directly. Figures~\ref{fig:pq-nshd}--\ref{fig:pq-avr} show the average NSHD,
power and AVR over the whole \(8\times8\) grid of \((\kappa,\kappa')\), with
the pairs of Table~\ref{tab:reachability-summary} marked. 

At the reported pairs, the procedure achieved high power and nearly perfect
causal ordering. Power ranged from \(0.9686\) to \(1\), and AVR never exceeded
\(0.0001\). Reachability recovery was more accurate for \(\alpha=1\) in every
dimension: the NSHD increased from \(0.0000\), \(0.0016\) and \(0.0027\) for
\(\alpha=1\) to \(0.0060\), \(0.0180\) and \(0.0307\) for \(\alpha=2\).
Remark~\ref{rem:threshold} explains this ordering. For a non-ancestral pair,
$\widehat F_{ij}$ is of order $U_i(n)^{-1}$, whereas the threshold
$X_{i,(n-k_n')}^{-1}$ is of order $U_i(n/k_n')^{-1}$. Their ratio therefore
grows like $(k_n')^{1/\alpha}$, up to a slowly varying factor, and this gap is
wider when $\alpha=1$.

\begin{figure}[!t]
\centering
\includegraphics[width=\textwidth]{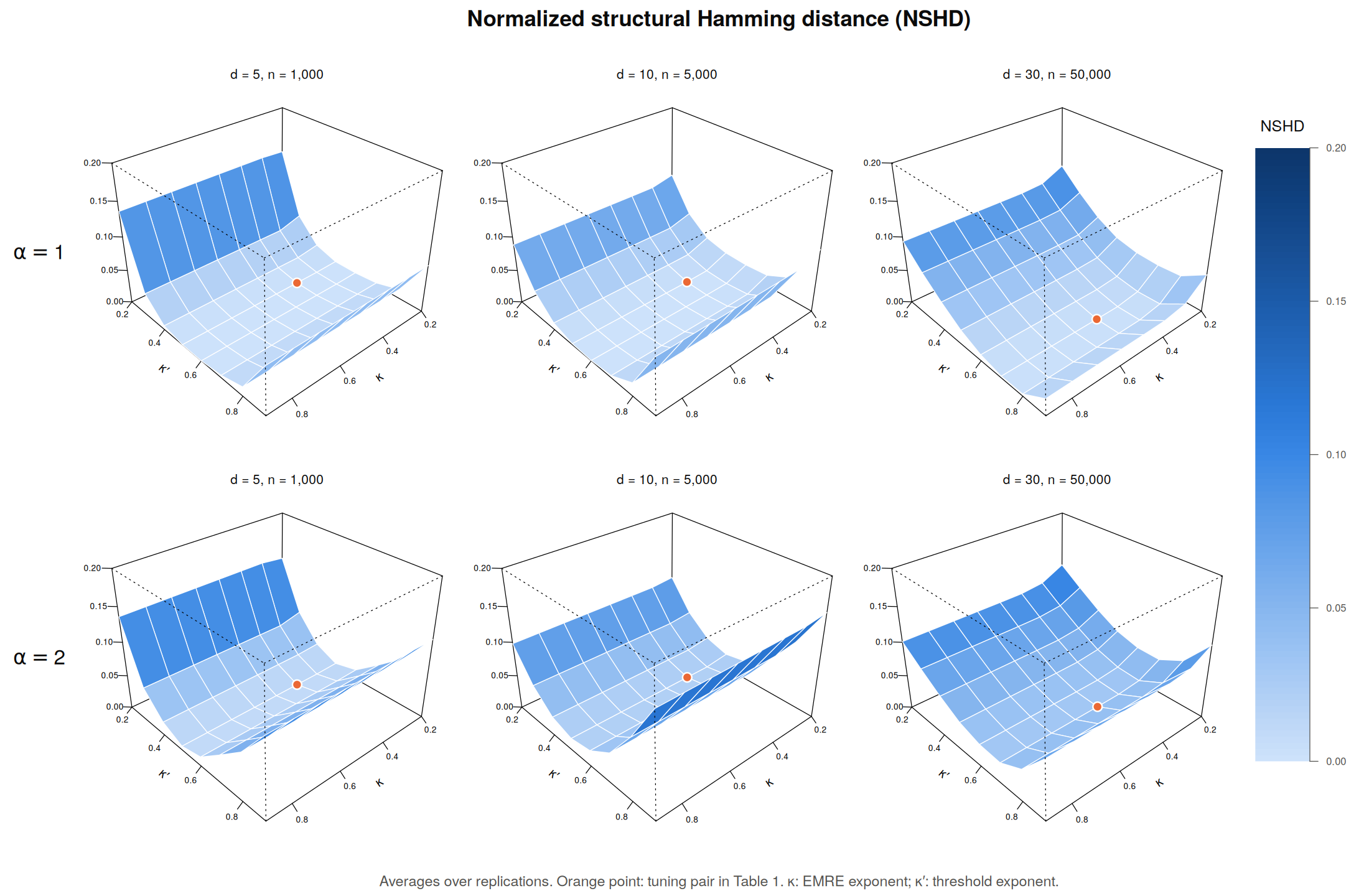}
\caption{Average normalized structural Hamming distance over the tuning grid,
for \(\alpha=1\) (top row) and \(\alpha=2\) (bottom row). Here \(\kappa\) is
the EMRE exponent, \(k_n=\lfloor n^{\kappa}\rfloor\), and \(\kappa'\) is the
threshold exponent, \(k_n'=\lfloor n^{\kappa'}\rfloor\). Each surface averages
50 replications, all panels share one colour scale, and the orange point marks
the pair reported in Table~\ref{tab:reachability-summary}.}
\label{fig:pq-nshd}
\end{figure}

\paragraph{Threshold exponent.}
The NSHD surfaces in Figure~\ref{fig:pq-nshd} form a valley in the
\(\kappa'\) direction, with one type of error on each side. For small
\(\kappa'\), the threshold is too low and non-ancestral pairs are declared
ancestral. At \(\kappa'=0.2\) and \(\kappa=0.5\), the average number of false
positives is \(2.7\), \(8.0\) and \(81.3\) for \(d=5\), \(10\) and \(30\) when
\(\alpha=1\), and no ancestral relation is missed. For large \(\kappa'\), the
threshold is too high and weak ancestral relations are missed. At
\(\kappa'=0.9\), false positives essentially vanish, but the average number of
false negatives is \(1.4\), \(7.4\) and \(21.8\).

The power surfaces in Figure~\ref{fig:pq-power} show the second effect
directly. For \(\alpha=1\), power is essentially one for
\(\kappa'\leq0.6\) and decreases afterwards; for \(\alpha=2\), the decrease
starts earlier, near \(\kappa'=0.4\) when \(d\leq10\). At \(\kappa'=0.9\), it
equals \(0.585\), \(0.683\) and \(0.931\) for \(\alpha=1\), and \(0.315\),
\(0.428\) and \(0.784\) for \(\alpha=2\). For \(\kappa'=0.9\) and \(n=1000\),
we have \(k_n'=501\), so the threshold is the reciprocal of roughly the sample
median of \(X_i\). Such a threshold exceeds the AIR of an ancestral pair that
is connected only through long paths of weak edges, since that AIR is built
from products of several weights in \((0.2,0.9)\).

\begin{figure}[!t]
\centering
\includegraphics[width=\textwidth]{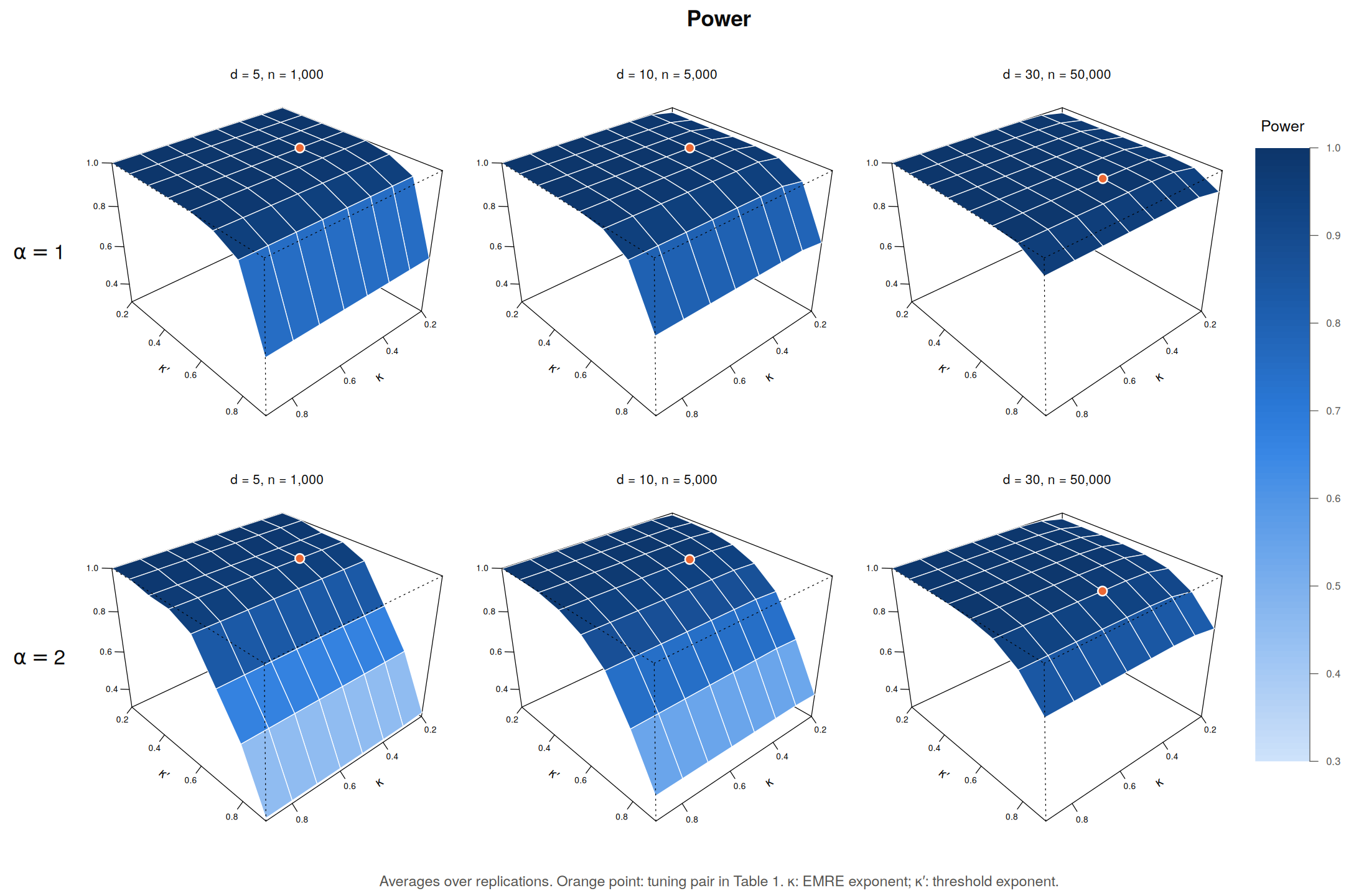}
\caption{Average power over the tuning grid, for \(\alpha=1\) (top row) and
\(\alpha=2\) (bottom row). Axes, replications, colour scale and orange points
are as in Figure~\ref{fig:pq-nshd}.}
\label{fig:pq-power}
\end{figure}

Between the two error regimes lies a band of low NSHD. For \(\alpha=1\) and
\(\kappa\geq0.4\), the NSHD is at most \(0.01\) for \(\kappa'\in[0.4,0.7]\)
when \(d=5\), for \(\kappa'\in[0.5,0.7]\) when \(d=10\), and for
\(\kappa'\in[0.6,0.8]\) when \(d=30\). The band moves to larger \(\kappa'\) as
\((d,n)\) grows. With more observations, a larger \(k_n'\) still leaves
\(n/k_n'\) large enough for the threshold to stay below the AIRs of ancestral
pairs, and the larger \(k_n'\) widens the gap \((k_n')^{1/\alpha}\) between
the threshold and the non-ancestral estimates. Because \(d\) and \(n\)
increase together in this design, the two effects cannot be separated here.

For \(\alpha=2\), the valley is narrower and shallower. The minimum NSHD is
\(0.006\), \(0.018\) and \(0.031\), attained at \(\kappa'=0.5\), \(0.5\) and
\(0.7\), and only \(d=5\) has grid points with NSHD at most \(0.01\). Both
error types increase. The gap \((k_n')^{1/\alpha}\) between the threshold and
the non-ancestral estimates is smaller, and the threshold, of order
\((n/k_n')^{-1/\alpha}\), decreases more slowly, so weak ancestral relations
are missed at smaller \(\kappa'\). For \(d=30\) and \(\alpha=2\), false
positives persist even at \(\kappa'=0.9\), with \(8.2\) on average. They
plausibly come from non-ancestral pairs whose common ancestors dominate the
private part of \(X_i\). For such pairs, the constant \(A\) in part~(a) of
Proposition~\ref{prop:EMRE-consistency} is small, and the EMRE approaches
zero slowly.

\begin{figure}[!t]
\centering
\includegraphics[width=\textwidth]{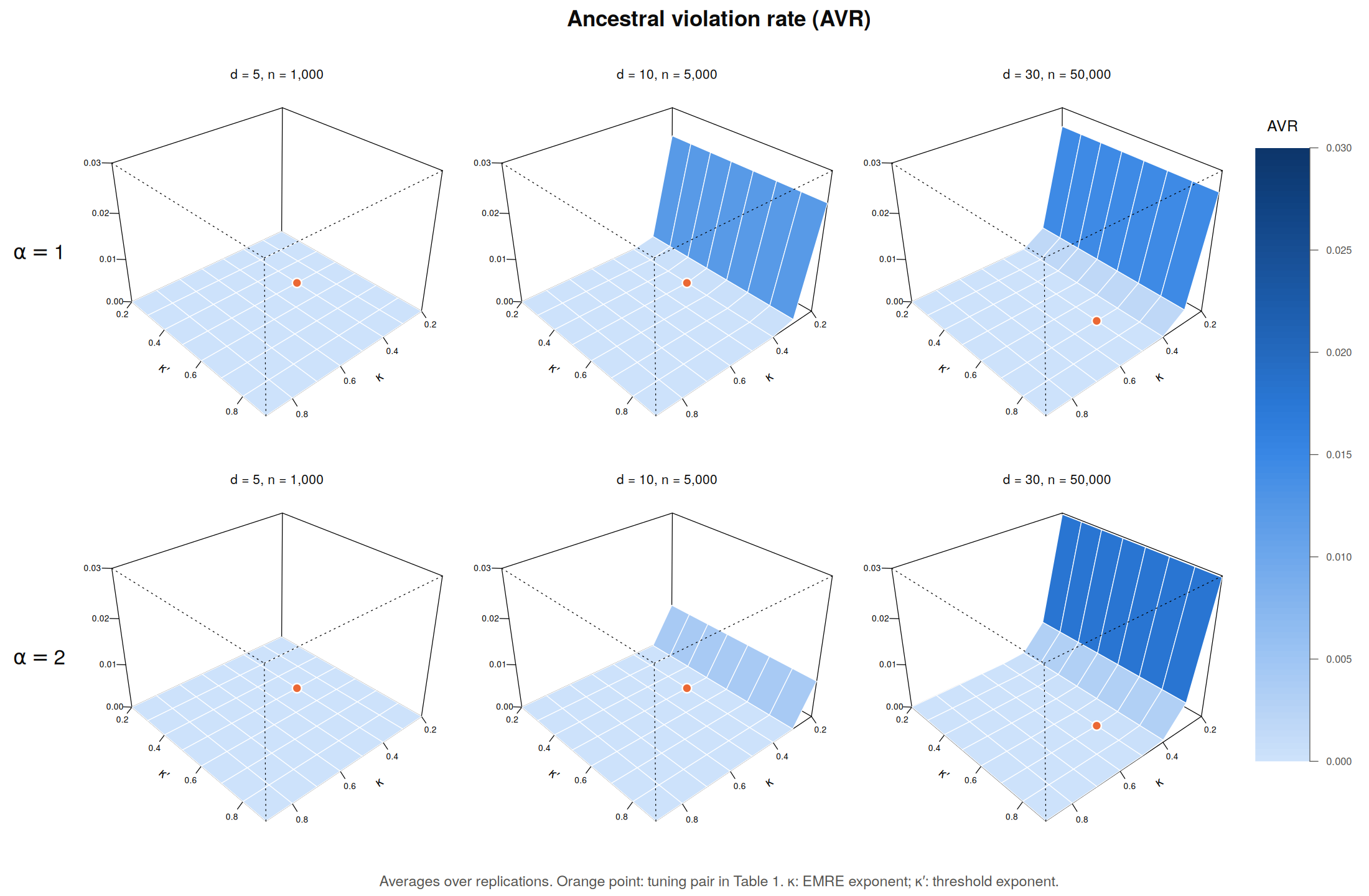}
\caption{Average ancestral violation rate over the tuning grid, for
\(\alpha=1\) (top row) and \(\alpha=2\) (bottom row). The estimated order
depends only on \(\kappa\), so each surface is constant in \(\kappa'\). Axes,
replications, colour scale and orange points are as in
Figure~\ref{fig:pq-nshd}.}
\label{fig:pq-avr}
\end{figure}

\paragraph{EMRE exponent.}
Beyond a moderate value of \(\kappa\), the NSHD and power surfaces are flat in
the \(\kappa\) direction. To four decimals, the NSHD does not change with
\(\kappa\) once \(\kappa\geq0.3\), \(0.4\) and \(0.5\) for \(d=5\), \(10\) and
\(30\), respectively, and over the same ranges the power changes by less than
\(10^{-4}\). Two properties of the EMRE explain this. The sets of extreme
observations are nested in \(k_n\), so \(\widehat F_{ij}(\kappa)\) is
nonincreasing in \(\kappa\). By Lemma~\ref{lem:ratio-lower-bound},
\(\widehat F_{ij}(\kappa)\geq F_{ij}\) for every \(\kappa\). Increasing
\(\kappa\) can therefore only lower the estimates. For ancestral pairs, they
remain bounded below by the AIR; for non-ancestral pairs, they move further
below the threshold.

The estimated order depends only on \(\widehat F(\kappa)\), so the AVR
surfaces in Figure~\ref{fig:pq-avr} are constant in \(\kappa'\). For \(d=5\),
the order is recovered without error at every \(\kappa\). For \(d=10\) and
\(30\), ordering errors are concentrated at the smallest exponents. At
\(\kappa=0.2\), where \(k_n=5\) and \(8\), the AVR equals \(0.024\) and
\(0.026\) for \(\alpha=1\), and \(0.008\) and \(0.030\) for \(\alpha=2\). With
so few extreme observations, the EMRE of a non-ancestral pair need not be
small, and EASE can place a descendant before one of its ancestors. For
\(\kappa\geq0.4\), the AVR is at most \(0.0005\) in every configuration.

\paragraph{Choice of tuning parameters.}
Taken together, the three figures suggest a moderate \(\kappa\), for which
both the estimated order and the reachability decisions are stable, and an
intermediate \(\kappa'\) that increases with the sample size. The pairs in
Table~\ref{tab:reachability-summary} follow this rule. They lie near, but not
always at, the grid minimum of the NSHD; for example, for \((d,n)=(10,5000)\)
and \(\alpha=1\), the minimum \(0.0004\) is attained at \(\kappa'=0.6\), for
instance at \((\kappa,\kappa')=(0.4,0.6)\). By
Remark~\ref{rem:threshold}, the two exponents may be chosen separately, as
they are here.

\section{Data Application: The upper Danube River basin}

\label{sec:danube-sensitivity}

The causal order of Section~\ref{sec:structure} and the reachability
estimator~\eqref{eq:ordered-reachability-estimator} depend on two intermediate
sequences: $k_n=\lfloor n^{\kappa}\rfloor$, the number of
largest observations of $X_i$ used by the EMRE, and $k_n'=\lfloor n^{\kappa'}\rfloor$,
which sets the threshold $X_{i,(n-k_n')}^{-1}$ of the reachability estimator.
The theory only requires $k_n,k_n'\to\infty$ with $k_n/n,k_n'/n\to0$, under
which the threshold satisfies~\eqref{eq:threshold-conditions}; this leaves the
exponents $\kappa,\kappa'\in(0,1)$ to the user. In this section we
examine how the estimated causal order and the estimated reachability matrix
change with $(\kappa,\kappa')$ on river discharge data whose causal structure
is known.

\subsection{Data and preprocessing}
\label{sec:danube-data}

We use the twelve gauging stations of the upper Danube basin analysed by
\citet{gnecco2021causal}, shown in Figure~\ref{fig:danube-network}. Eleven of
them belong to the 31 stations studied by \citet{asadi2015extremes}; the
twelfth, labelled 32, lies on the Inn close to the Austrian border. The data are
daily mean discharges (m$^3$/s) recorded by the Bavarian Environmental Agency and
distributed with the replication material of \citet{gnecco2021causal}. The
river network has eleven edges, and its transitive closure contains $18$
ancestral pairs among the $d(d-1)=132$ ordered pairs of distinct stations.

\begin{figure}[htbp]
  \centering
  \includegraphics[width=0.72\textwidth]{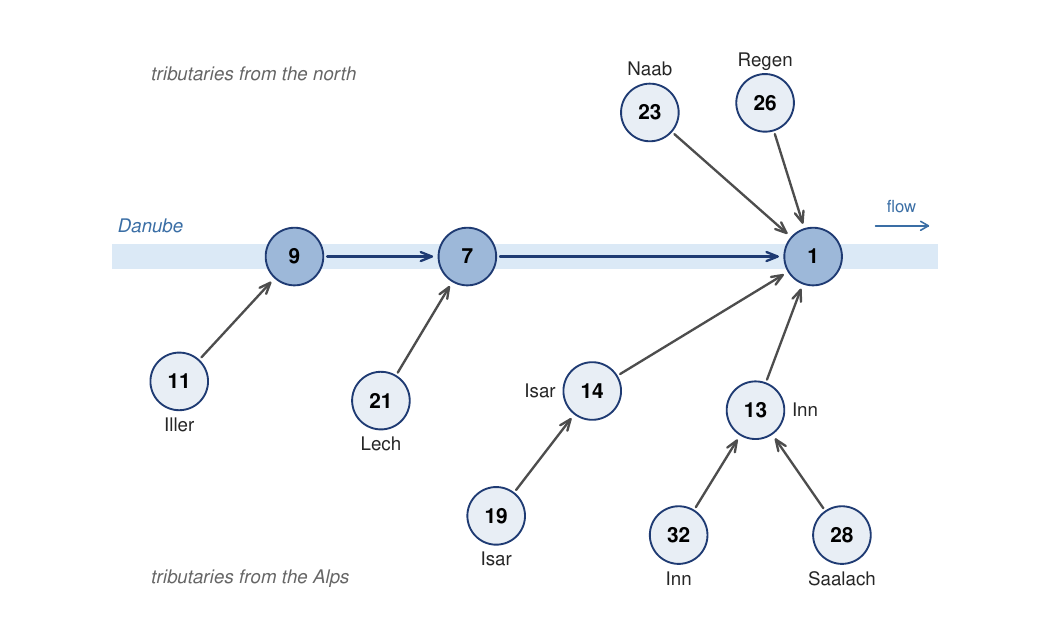}
  \caption{Schematic of the river network formed by the twelve gauging
    stations (not to scale). Dark nodes lie on the Danube, which flows from
    left to right; tributaries from the north are drawn above it and alpine
    tributaries below. Arrows give the direction of flow. Station~1 is the most
    downstream station, and station~32 lies on the Inn upstream of
    station~13.}
  \label{fig:danube-network}
\end{figure}

Daily discharges are strongly dependent in time, and a single flood is
recorded at several stations on consecutive days. Following
\citet{asadi2015extremes}, we restrict attention to June, July and August of
1960--2010, the months in which most floods in the basin occur, and extract
flood events as follows. Each daily observation is replaced by its rank within
the summer series of its station. The summer day with the largest rank across
the twelve stations is selected, provided that the nine-day window centred at
it does not overlap a window that has already been used; the corresponding
event is the vector of station-wise maxima over the window. The window is then
removed and the procedure is repeated until no admissible window remains.
Windows may extend into the last days of May or the first days of September.
Applied to the 31 stations of \citet{asadi2015extremes}, this procedure returns
$428$ events, the number in the declustered data distributed with the R package
\texttt{graphicalExtremes}, and $87\%$ of these events coincide exactly with
the published ones. For the twelve stations it yields $n=430$ flood events,
which form the data set analysed below.

The margins are standardized separately at each station. For station $i$,
let $x_{i,(1)}\le x_{i,(2)}\le\cdots\le x_{i,(n)}$ be the order statistics of
the event values $x_{i1},\ldots,x_{in}$, let $m<n$, and let
$m_i:=\#\{\ell:x_{i\ell}>x_{i,(n-m)}\}$ be the number of
exceedances of $x_{i,(n-m)}$, so that $m_i=m$ in the absence of ties.
Following routine practice in extreme value analysis, the distribution
function $G_i$ of $X_i$ is estimated semiparametrically by
\[
  \widehat G_i(x)=
  \begin{cases}
    \dfrac{1}{n}\displaystyle\sum_{\ell=1}^{n}\mathbf 1_{\{x_{i\ell}\le x\}},
      & x\le x_{i,(n-m)},\\[4mm]
    1-\dfrac{m_i}{n}\Bigl(1+\hat\gamma_i\,\dfrac{x-x_{i,(n-m)}}{\hat\sigma_i}\Bigr)^{-1/\hat\gamma_i},
      & x>x_{i,(n-m)},
  \end{cases}
\]
where $\hat\gamma_i$ and $\hat\sigma_i$ are the maximum likelihood estimates of
the shape and scale parameters of a generalized Pareto distribution fitted to
the $m_i$ exceedances $x_{i\ell}-x_{i,(n-m)}>0$. Both pieces equal $1-m_i/n$ at
$x_{i,(n-m)}$, so $\widehat G_i$ is continuous and nondecreasing. We take
$m=50$, which gives $m_i=50$ at ten stations and $m_i=49$ at the other two,
and set $Z_{i\ell}:=\{1-\widehat G_i(x_{i\ell})\}^{-1/2}$, so that every
station has approximately standard Pareto margins with tail index~$2$. The
fitted shape parameters range from $-0.22$ to $0.29$; four of the twelve are
negative and one is essentially zero; the tails of several stations are therefore close to the
boundary of the heavy-tailed regime assumed by the theory, in line with the
regional estimates of \citet{asadi2015extremes}. Because the transformation
applies the same power to every margin, all decisions below are unchanged if
tail index~$2$ is replaced by any other value.

\subsection{Design of the sensitivity study}
\label{sec:danube-design}

Both exponents range over the grid
\[
  \mathcal K_{0.01}=\{0.20,0.21,\ldots,0.90\},
\]
giving $71^2=5041$ combinations $(\kappa,\kappa')$. With $n=430$, the number
$k_n=\lfloor n^{\kappa}\rfloor$ ranges from $3$ to $234$ and takes $58$
distinct values, so neighbouring values of $\kappa$ (and of $\kappa'$) often
lead to the same estimator. For each $\kappa$, the EMRE matrix
$\widehat F(\kappa)$ is computed from $Z$, with ties among the $k_n$ largest
observations broken at random, and EASE (Algorithm~\ref{alg:ease}) returns an
order $\widehat\pi_\kappa$. For each $\kappa'$, the reachability matrix
$\widehat R(\kappa,\kappa')$ is obtained from
\eqref{eq:ordered-reachability-estimator} with threshold
$X_{i,(n-k_n')}^{-1}$, computed on the standardized scale.

The order is assessed by its ancestral violation rate
$\operatorname{AVR}(\widehat\pi_\kappa)$, the proportion of the $18$ ancestral
pairs placed in the wrong direction. The reachability matrix is assessed on its
$132$ off-diagonal entries by
\[
  \operatorname{Power}=\frac{\operatorname{TP}}{\operatorname{TP}+\operatorname{FN}}
  \qquad\text{and}\qquad
  \operatorname{NSHD}=\frac{\operatorname{FP}+\operatorname{FN}}{d(d-1)} .
\]
The empty graph, which declares no ancestral relation, has
$\operatorname{NSHD}=18/132=0.136$; an estimator is useful only if it improves
on this benchmark.

\subsection{Results}
\label{sec:danube-results}

\paragraph{Causal order.}
The estimated order depends on $\kappa$ only, and Table~\ref{tab:danube-order}
and Figure~\ref{fig:danube-avr} show that it changes in four steps. For
$0.35\le\kappa\le0.55$, that is, whenever the EMRE uses between $8$ and $28$
events, EASE returns the order
$23\succ26\succ28\succ32\succ19\succ14\succ21\succ11\succ9\succ7\succ13\succ1$,
which has no ancestral violations. With fewer events the order is unstable and
places the outlet, station~$1$, before several of its ancestors. With $29$ or
more events, the order makes a single error that does not change up to
$k_n=234$: station~$1$ is placed before the Danube stations $9$ and~$7$,
giving $\operatorname{AVR}=2/18$.

\begin{table}[t]
  \centering
  \caption{Estimated causal order as a function of $\kappa$, with
    $k_n=\lfloor n^{\kappa}\rfloor$ and $n=430$.}
  \label{tab:danube-order}
  \begin{tabular}{@{}cccl@{}}
    \toprule
    $\kappa$ & $k_n$ & AVR & estimated order $\widehat\pi_\kappa$ \\
    \midrule
    $0.20$--$0.26$ & $3$--$4$    & $0.333$ & $32\succ23\succ11\succ19\succ14\succ21\succ13\succ1\succ26\succ28\succ7\succ9$ \\
    $0.27$--$0.34$ & $5$--$7$    & $0.167$ & $26\succ23\succ28\succ32\succ19\succ14\succ21\succ11\succ13\succ1\succ7\succ9$ \\
    $0.35$--$0.55$ & $8$--$28$   & $0$     & $23\succ26\succ28\succ32\succ19\succ14\succ21\succ11\succ9\succ7\succ13\succ1$ \\
    $0.56$--$0.90$ & $29$--$234$ & $0.111$ & $23\succ26\succ28\succ32\succ19\succ14\succ21\succ11\succ13\succ1\succ9\succ7$ \\
    \bottomrule
  \end{tabular}

  \smallskip
  \footnotesize For $\kappa\in[0.20,0.26]$ and $[0.27,0.34]$ the order shown is
  the one at $\kappa=0.20$ and $\kappa=0.30$; neighbouring values of $\kappa$ in
  these ranges give orders with the same violation rate.
\end{table}

\begin{figure}[t]
  \centering
  \includegraphics[width=0.72\textwidth]{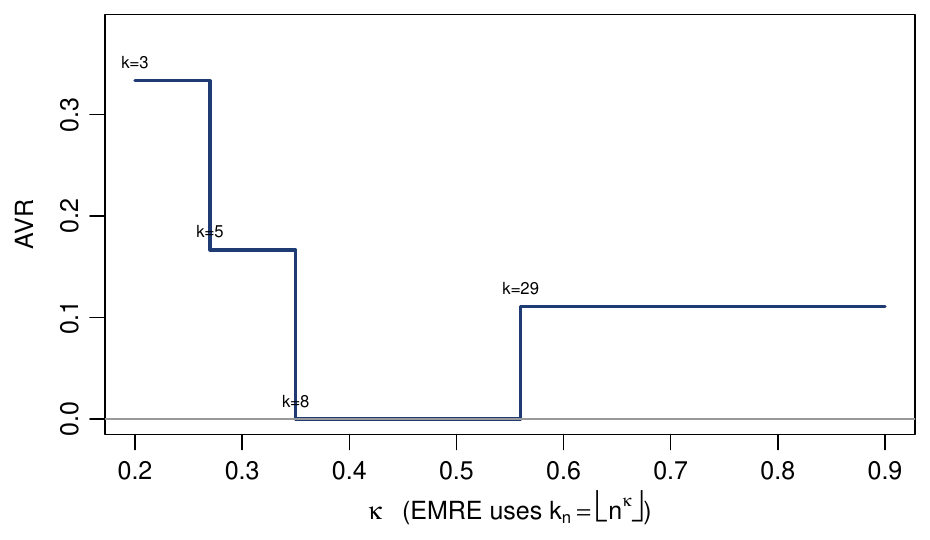}
  \caption{Ancestral violation rate of the EASE order as a function of
    $\kappa$. Labels give $k_n$ at each change.}
  \label{fig:danube-avr}
\end{figure}

\paragraph{Reachability.}
Figure~\ref{fig:danube-heatmaps} shows NSHD and power over the whole grid, and
Table~\ref{tab:danube-selected} lists selected settings. The threshold exponent
$\kappa'$ governs the trade-off between false positives and missed relations.
For $\kappa'\le0.45$ the threshold $X_{i,(n-k_n')}^{-1}$ is small and most
pairs that EASE places in ancestral order are declared ancestral: every such
setting is worse than the empty graph ($\operatorname{NSHD}\ge0.144$), although
the power is at least $0.78$ whenever $\kappa\ge0.35$. For $\kappa'\ge0.80$
and $\kappa\ge0.35$ the threshold is so large that at most two relations are
detected, all of them correctly. Between these
extremes, $1724$ of the $5041$ settings ($34.2\%$) improve on the empty graph,
all with $0.50\le\kappa'\le0.90$; $640$ of them also have a valid order.

The best results form a band in $\kappa'$ rather than a point. For
$\kappa\in[0.35,0.55]$, where the order is valid, NSHD lies between $0.098$
and $0.114$ for every $\kappa'\in[0.57,0.68]$, that is, for $k_n'$ between $31$
and $61$; for $\kappa\ge0.56$, where the order makes one error, it lies between
$0.106$ and $0.121$ over the same band. Over the wider band
$\kappa'\in[0.56,0.75]$, no setting with a valid order is worse than the empty
graph. The
smallest value, $\operatorname{NSHD}=0.098$, is attained for every
$\kappa\in[0.35,0.55]$ at six values of $\kappa'$ between $0.59$ and $0.65$
($k_n'$ between $35$ and $51$), with errors that shift from false positives to
missed relations as $\kappa'$ grows: $13$ true positives, $8$ false positives
and $5$ missed relations (power $0.72$) at $\kappa'=0.59$; $10$, $5$ and $8$
(power $0.56$) for $\kappa'\in[0.60,0.63]$; and $8$, $3$ and $10$ (power
$0.44$) at $\kappa'=0.65$. If a power of
at least $0.8$ is required, the best setting has $\operatorname{NSHD}=0.129$
($15$ true positives, $14$ false positives, $3$ missed relations) at
$\kappa\in[0.35,0.40]$ and $\kappa'\in[0.50,0.52]$. Within the band of good
values, the results hardly depend on $\kappa$: once the order is valid, the
reachability estimate is driven by the threshold.

The best settings have $\kappa'>\kappa$, so that the threshold is based on more
order statistics than the EMRE itself. This agrees with the rate theory behind
Proposition~\ref{prop:reachability-consistency}. For a non-ancestral pair the EMRE is of order
$U_i(n)^{-1}$, while the threshold is of order $U_i(n/k_n')^{-1}$, larger by a
factor of order $(k_n')^{1/\alpha}$. A small $k_n'$ leaves too little room
between the two and produces false positives; a large $k_n'$ pushes the
threshold above the smaller ancestral impulse responses and produces missed
relations.

\begin{figure}[t]
  \centering
  \includegraphics[width=\textwidth]{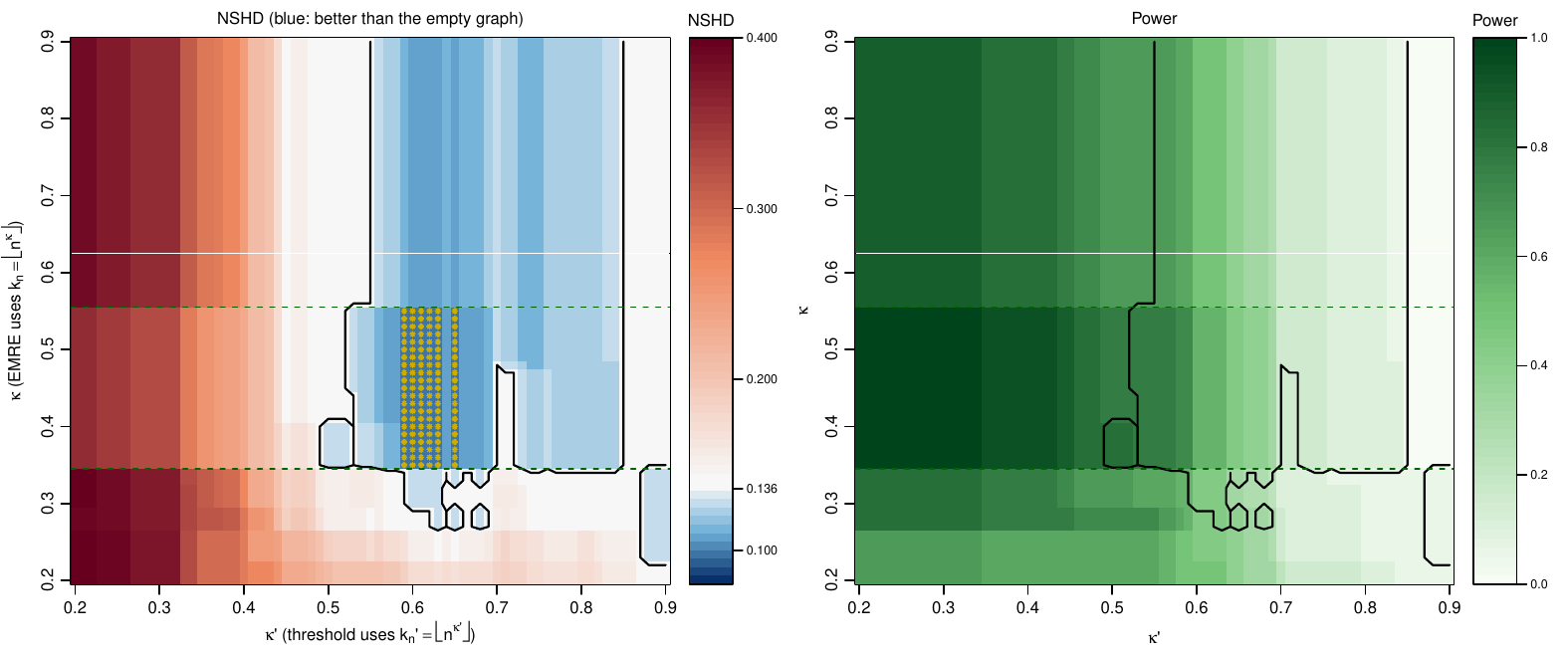}
  \caption{NSHD (left) and power (right) of the reachability estimator over
    $(\kappa,\kappa')\in\mathcal K_{0.01}^2$. Black contours enclose the
    settings that improve on the empty graph ($\operatorname{NSHD}<0.136$);
    dashed lines mark the range $0.35\le\kappa\le0.55$ in which the EASE order
    has no ancestral violations; stars mark the minimum NSHD.}
  \label{fig:danube-heatmaps}
\end{figure}

\begin{table}[t]
  \centering
  \caption{True positives, false positives and missed relations
    (TP/FP/FN) and NSHD for selected $(\kappa,\kappa')$. The empty graph has
    $\operatorname{NSHD}=0.136$.}
  \label{tab:danube-selected}
  \begin{tabular}{@{}cc cc cc cc@{}}
    \toprule
    & & \multicolumn{2}{c}{$\kappa=0.30$ ($k_n=6$)}
      & \multicolumn{2}{c}{$\kappa=0.40$ ($k_n=11$)}
      & \multicolumn{2}{c}{$\kappa=0.70$ ($k_n=69$)} \\
    \cmidrule(lr){3-4}\cmidrule(lr){5-6}\cmidrule(l){7-8}
    $\kappa'$ & $k_n'$ & TP/FP/FN & NSHD & TP/FP/FN & NSHD & TP/FP/FN & NSHD \\
    \midrule
    $0.40$ & $11$  & $14/32/4$  & $0.273$ & $17/30/1$  & $0.235$ & $15/30/3$  & $0.250$ \\
    $0.50$ & $20$  & $12/14/6$  & $0.152$ & $15/14/3$  & $0.129$ & $12/12/6$  & $0.136$ \\
    $0.59$ & $35$  & $10/10/8$  & $0.136$ & $13/8/5$   & $0.098$ & $11/8/7$   & $0.114$ \\
    $0.63$ & $45$  & $8/7/10$   & $0.129$ & $10/5/8$   & $0.098$ & $9/5/9$    & $0.106$ \\
    $0.70$ & $69$  & $2/5/16$   & $0.159$ & $3/3/15$   & $0.136$ & $3/2/15$   & $0.129$ \\
    $0.80$ & $127$ & $1/2/17$   & $0.144$ & $2/0/16$   & $0.121$ & $2/0/16$   & $0.121$ \\
    \midrule
    \multicolumn{2}{@{}l}{AVR} & \multicolumn{2}{c}{$0.167$} & \multicolumn{2}{c}{$0$} & \multicolumn{2}{c}{$0.111$} \\
    \bottomrule
  \end{tabular}
\end{table}

\paragraph{Where the errors occur.}
Figure~\ref{fig:danube-reachability} compares the estimated and true
ancestral relations at $\kappa=0.40$ for $\kappa'=0.50$ and $\kappa'=0.59$.
At $\kappa'=0.50$, $12$ of the $14$ false positives run from the two Isar
stations ($19$ and $14$) or the Lech ($21$) to stations in neighbouring
catchments: the Lech, the Iller ($11$), the Danube stations $9$ and $7$, and
the Inn ($13$); the remaining two run from the Danube stations $9$ and $7$ to
the Inn. These catchments are adjacent, and their extreme
discharges are driven by the same summer storms, a common cause that the
heavy-tailed structural causal model does not include. All three missed
relations end at the outlet ($11$, $23$ and $28\to1$). At $\kappa'=0.59$, eight
false positives remain, all running from the Isar stations or the Lech to
neighbouring catchments, and all five missed relations end at station~$1$
($11$, $19$, $23$, $26$ and $28\to1$). A flood in a single
tributary moves the outlet only slightly on the standardized scale, so these
ancestral impulse responses are small and are the first to fall below the
threshold.

\begin{figure}[t]
  \centering
  \includegraphics[width=\textwidth]{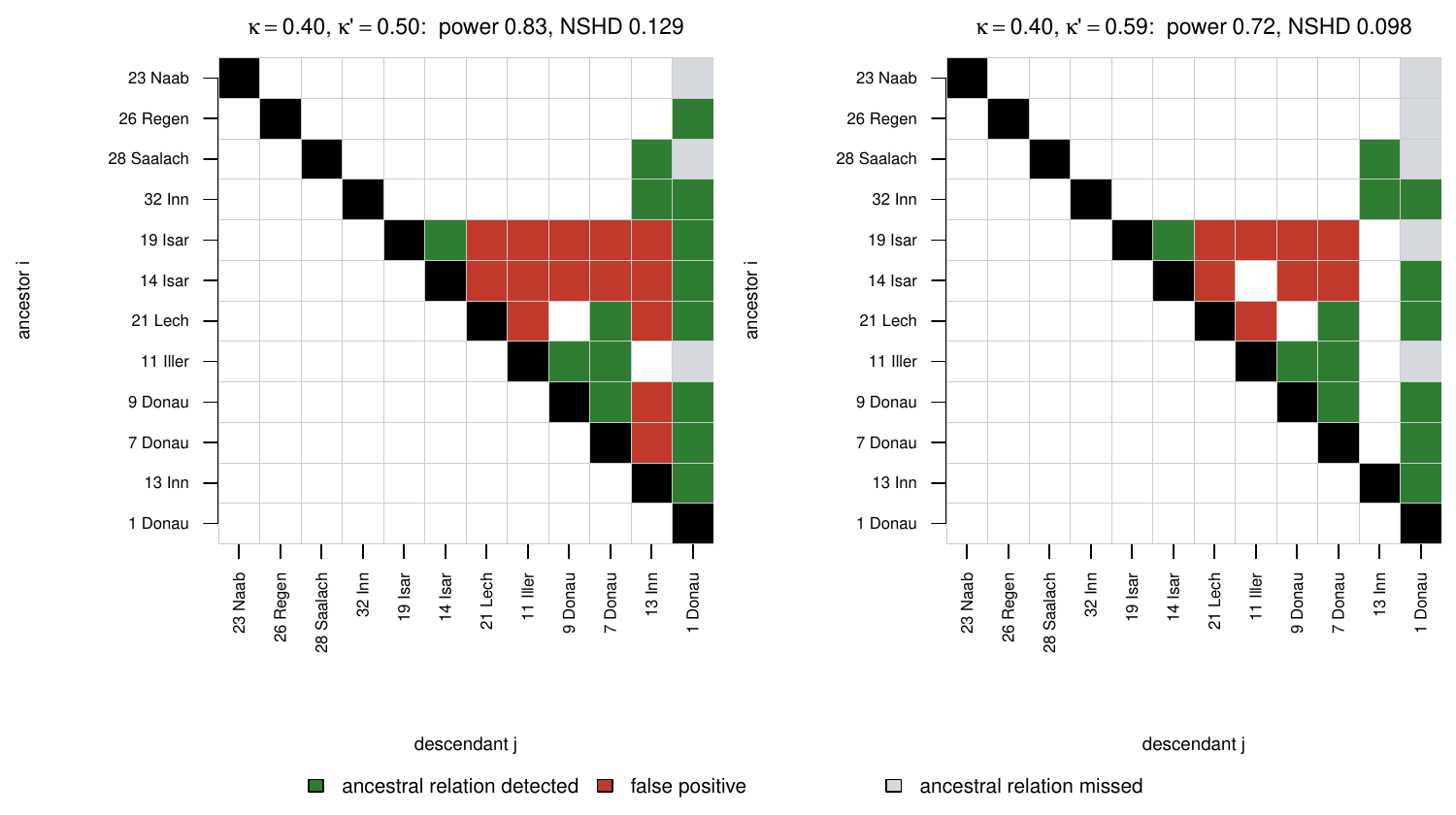}
  \caption{Estimated against true ancestral relations at $\kappa=0.40$, with
    $\kappa'=0.50$ (left) and $\kappa'=0.59$ (right). Rows are potential
    ancestors and columns potential descendants, both in the estimated causal
    order.}
  \label{fig:danube-reachability}
\end{figure}

\subsection{Discussion}
\label{sec:danube-discussion}

Three conclusions emerge. First, the causal order is correct for a range of
$k_n$ rather than a single value, here $8\le k_n\le28$, and degrades
gracefully for larger $k_n$, with one error that remains fixed up to
$k_n=234$. Second, the threshold exponent matters more than the EMRE exponent:
the good settings form a band $0.57\le\kappa'\le0.68$ that is almost
independent of $\kappa$, and choosing $\kappa'$ larger than $\kappa$ agrees
with the rate theory. Third, the remaining errors have a physical
explanation, namely shared storms in adjacent catchments and small impulse
responses into the outlet.

Two caveats apply. First, the minimum NSHD of $0.098$ is selected from $5041$
settings using the true graph, and is therefore optimistic; the band of values
reported above is a more reliable summary. A data-driven choice of
$(\kappa,\kappa')$, for instance one that maximizes the stability of
$\widehat R(\kappa,\kappa')$ across subsamples, is left for future work.
Second, ordering the stations by their mean event discharge also yields a valid
order on these data, because downstream stations carry more water. The EMRE
does not use this information, since it operates on standardized margins, but
the comparison shows that the causal order alone is a modest test on river
networks; the reachability results are the more informative part of the
analysis.

\section{Conclusion}
\label{sec:conclusion}

This paper introduced the extremal minimum ratio estimator of the ancestral
impulse responses of a heavy-tailed homogeneous structural causal model and
showed that it is consistent under regular variation alone. Under additional
conditions on the lower tail of the innovations, we derived its exact rates of
convergence and limiting distributions, for non-ancestral pairs in every
HT-HSCM and for ancestral pairs in $\ell_p$-aggregation models with
$0<p<\infty$; in the max-linear case $p=\infty$, the estimator recovers the
AIR exactly with probability tending to one. These rates lead to a consistent
estimator of the causal order and to a rate-based estimator of the
reachability matrix of the underlying DAG that requires no knowledge of the
tail index. In the application to the upper Danube basin, the procedure
recovered a valid causal order over a range of tuning parameters, and most of
its remaining errors stem from a common cause, shared summer storms, that lies
outside the model.

Several extensions are natural. The homogeneity assumption could be relaxed,
and the robustness of the estimator to departures from the model, in
particular to latent common causes, remains to be studied. The estimator could
also be modified to accommodate innovations that are regularly varying in both
tails, and its performance could be compared with that of the causal tail
coefficient of \citet{gnecco2021causal} and of the scaling-based methods of
\citet{kluppelberg2021estimating} and \citet{krali2026causal}.

\section{Acknowledgements}
I would like to thank Dr.Shuyang Bai for his insightful comments on the paper.

\section*{Funding}
No funding was received for this work.

\bibliographystyle{apalike}
\bibliography{ref}
\end{document}